\documentclass[sigconf,9pt]{acmart}

\usepackage[subtle,title=tight]{savetrees} 
\usepackage[english]{babel}
\usepackage{blindtext}

\setcopyright{cc}
\setcctype{by}
\copyrightyear{2026}
\acmYear{2026}
\acmConference[SIGCOMM '26]{ACM SIGCOMM 2026 Conference}{August 17--21, 2026}{Denver, CO, USA}
\acmBooktitle{ACM SIGCOMM 2026 Conference (SIGCOMM '26), August 17--21, 2026, Denver, CO, USA}
\acmDOI{10.1145/3789240.3829150}\acmISBN{979-8-4007-2467-1/26/08}

\usepackage[T1]{fontenc}
\usepackage{newtxtext,newtxmath}
\usepackage{float}
\usepackage{tikz}
\usepackage{amsmath}
\usepackage{xspace}
\usepackage{tabularx}

\usepackage{booktabs}
\usepackage{graphicx}
\usepackage{subcaption}

\usepackage{algorithm}
\usepackage{algorithmicx}
\usepackage{algpseudocode}
\algrenewcommand\algorithmicindent{1em}
\usepackage{booktabs}
\usepackage{siunitx}
\usepackage{comment}
\usepackage{etoolbox}
\usepackage{makecell}

\usepackage[newfloat]{minted}

\usepackage[table,xcdraw]{xcolor}
\usepackage{tikz}
\newcommand*\circled[1]{\tikz[baseline=(char.base)]{
            \node[shape=circle,draw,inner sep=0.5pt,font=\small] (char) {#1};}}
\usetikzlibrary{fit, shapes.geometric, arrows}
\usepackage{nicematrix}
\usepackage{caption}

\definecolor{bgcolor}{rgb}{0.99,0.99,0.99}

\usepackage{listings}

\definecolor{codegray}{rgb}{0.5,0.5,0.5}
\definecolor{codegreen}{rgb}{0,0.6,0}
\definecolor{backcolor}{rgb}{0.99,0.99,0.99}
\makeatletter

\usetikzlibrary{fit}
\usepackage{nicematrix}
\usepackage{caption}

\usepackage{enumitem}
\usepackage{multirow}
\setitemize{noitemsep,topsep=0pt,parsep=0pt,partopsep=0pt,itemindent=.5em,leftmargin=0.5em}
\setenumerate{noitemsep,topsep=0pt,parsep=0pt,partopsep=0pt,itemindent=.5em,leftmargin=0.5em}

\newcommand{\para}[1]{%
  \par\addvspace{0.0\baselineskip}%
  \noindent\textbf{#1}\hspace{0.6em}%
}

\newcommand{\captionfonts}{\bf \small}
\usepackage[normalem]{ulem}

\makeatletter
\def\thm@space@setup{%
  \thm@preskip=2pt
  \thm@postskip=2pt
}
\makeatother
\usepackage{etoolbox}

\makeatletter  
\long\def\@makecaption#1#2{%
	\vskip\abovecaptionskip
	\sbox\@tempboxa{{\captionfonts #1: #2}}%
	\ifdim \wd\@tempboxa >\hsize
	{\captionfonts #1: #2\par}
	\else
	\hbox to\hsize{\hfil\box\@tempboxa\hfil}%
	\fi
	\vskip\belowcaptionskip}
\makeatother   

\usepackage{pifont}
\newcommand{\xmark}{\ding{55}}%

\newcommand{\showComments}{yes}
\newcommand{\note}[2]{
    \ifthenelse{\equal{\showComments}{yes}}{\textcolor{#1}{#2}}{}
}

\newcommand{\name}{MonkeyTree\xspace}
\newcommand{\sys}{MonkeyTree\xspace}

\newcommand{\exclude}[1]{}

\newtheoremstyle{compact}
 {1.5pt}      
 {1.5pt}      
 {\itshape} 
 {\parindent}         
 {\scshape}
 {.}        
 {.5em}     
 {}         
\theoremstyle{compact}
\newtheorem{theorem}{Theorem}
\newtheorem{corollary}{Corollary}

\makeatletter
\newcommand\footnoteref[1]{\protected@xdef\@thefnmark{\ref{#1}}\@footnotemark}
\makeatother

\hypersetup{
      breaklinks=true,
      colorlinks=true, 
      urlcolor=blue  
    }
\usepackage{xurl}
\begin{document}

\title{\name: Near-Minimal Congestion\\for Multi-tenant Training via Migration}

\author{Anton A. Zabreyko$^\dagger$ \quad
        Weiyang Wang$^\dagger$ \quad
        Manya Ghobadi$^{\dagger*}$}
\affiliation{%
  \institution{$^\dagger$MIT \qquad $^*$Systalyze}
  \country{}
}

\newcommand{\paperauthors}{Anton A. Zabreyko, Weiyang Wang, Manya Ghobadi}

\begin{abstract}
We present \sys, the first system to mitigate network congestion in multi-tenant GPU clusters through job-migration based \textit{defragmentation} rather than network-layer techniques. As cloud operators co-locate ML training jobs on shared, oversubscribed networks, congestion degrades training throughput for over a third of jobs. Prior approaches either rely on routing and flow scheduling—which we show have fundamental limits when traffic exceeds capacity, or require costly full-bisection bandwidth topologies with packet spraying.

\sys exploits characteristics of ML training traffic: ring-based collectives generate exactly one cross-rack flow per rack a job spans, making congestion-free placements \textit{achievable}. The sparse constraint structure admits abundant valid configurations, making them \textit{easy to reach} with few migrations. Once reached, low fragmentation is \textit{self-reinforcing}, as new arrivals disturb only a few racks. \sys formulates defragmentation as an integer linear program that minimizes worker movements, subject to per-rack fragmentation bounds. We prove a tight bound showing any placement can be defragmented to at most two cross-rack fragments per ToR, and extend the formulation to hybrid parallelism with multiple rings per server. Migration is implemented via in-memory checkpoint-and-restore over RDMA, incurring only $9.02$ seconds of system overhead end-to-end per worker. We evaluate \sys using a custom simulator modeling clusters of up to 2,048 H200 GPUs and prototype on a five-node A100 testbed. \sys improves average job completion time by 14\% over the next best baseline on a cluster of 1,024 GPUs with a 4:1 oversubscription. With a high 16:1 oversubscription ratio and 2,048 GPUs, \name keeps p99 job completion time within 5\% of ideal. 
\end{abstract} 


\settopmatter{printacmref=true, printccs=true, printfolios=false}

\renewcommand{\shortauthors}{Zabreyko et al.}
\renewcommand{\shorttitle}{\name}
\renewcommand{\authors}{\paperauthors}

\maketitle

\section{Introduction}

Multi-tenant GPU clusters have become a common deployment model for machine learning (ML) training workloads in the cloud~\cite{mlaas, crux, sglb, harmonics, cassini, antman}. In such environments, network congestion is a persistent challenge: concurrent jobs compete for shared bandwidth, degrading training throughput of 36.1\% of the jobs~\cite{crux}. The problem is compounded by the fact that operators traditionally build oversubscribed datacenter networks to reduce cost, further limiting available capacity during traffic bursts~\cite{meta_rdma, ncclx, hpn}. \looseness=-1

Prior work has tackled this problem from two directions. The first is to leverage routing~\cite{crux, sglb, foresight, harmonics} and flow-scheduling~\cite{cassini, mltcp} techniques to mitigate congestion by selecting efficient paths and time-multiplexing network resources across jobs. However, we demonstrate that these approaches have fundamental limits: when aggregate traffic exceeds network capacity, flow collisions become inevitable even with optimal routing and flow scheduling~(\S\ref{sec:motivation}). At the other extreme, industry leaders are pushing packet spraying over full-bisection bandwidth topologies as the ultimate solution~\cite{cisco-nvidia-spectrumx, cisco_packet, broadcom_packet, ethernet_packet}. This network architecture provides near-ideal performance, but at substantial cost. Achieving full-bisection bandwidth requires a large number of network switches, and packet spraying demands specialized switches and NICs~\cite{supernic}. Cloud providers must therefore trade performance against infrastructure expense.

We analyze the traffic patterns of multi-tenant DNN training jobs and show that traffic is dominated by \textit{sparse} elephant flows, in which each GPU communicates with only one other GPU for the majority of its traffic. This sparsity suggests network congestion is \textit{rare}, even on oversubscribed networks. In practice, however, jobs are often \textit{fragmented} across servers in different racks. Such fragmentation persists under high load despite locality-aware job schedulers~(\S\ref{sec:motivation}), causing the flow count on ToR-to-spine links to spike and increasing network load by 4.85$\times$ relative to optimal placement~(\S\ref{sec:eval_frag}). This leads to congestion despite sparse per-job traffic. Existing mitigation approaches address this only through networking techniques~\cite{crux,cassini,mltcp,sglb,harmonics,foresight}, but do not tackle the root cause, leaving clusters persistently fragmented. 

We argue that the most effective way to address congestion in multitenant GPU clusters lies not in optimizing network transfers but in \textit{job placement}. Based on this insight, we propose \sys, a system that continuously maintains a low-fragmentation state through small, proactive job \textit{migrations}—moving workers of the same job to concentrate them on the same racks. This \textit{defragmentation} reduces competing network traffic at its root cause, rather than reacting to congestion after it arises.

Using job migration to reduce congestion in GPU clusters remains unexplored, a legacy of CPU-centric datacenters. However, CPU workloads generate arbitrary traffic matrices, making congestion-free placements both ill-defined and rare. Reaching a defragmented placement that minimizes network traffic requires migrating a large fraction of workers, yet contention persists because no placement isolates traffic from workers entirely~(\S\ref{sec:cpu}). We argue that GPU training traffic is fundamentally different, and migration-based defragmentation is not only practical but highly effective for multi-tenant ML clusters, requiring only a small number of targeted migrations to eliminate congestion. \sys rests on three key insights about the structure of ML training traffic: \looseness=-1

First, congestion-free states are \textit{achievable} in GPU clusters.  The majority of ML training traffic consists of ring-based collectives that generate exactly one flow per rack a job spans. This property transforms the problem: rather than minimizing network traffic, we need only keep each rack's outgoing flows below its uplink count. Any placement satisfying this constraint achieves full path isolation.

Secondly, congestion-free states are \textit{abundant}. Unlike CPU workloads, where low-traffic placements are rare, the sparse constraint structure for GPU training admits many valid congestion-free configurations. Worker ordering and placement within racks do not affect traffic patterns, and each rack's constraint is independent. These degrees of freedom make congestion-free states easy to reach from any starting point.

Finally, maintaining low fragmentation is \textit{self-reinforcing}. Once a cluster reaches a valid state, maintaining it requires minimal effort. A new job violates constraints at only a few racks, and because valid placements are abundant, a small number of migrations restores compliance. \looseness=-1

Building on these insights, \sys employs a centralized controller and distributed daemons that operate alongside existing job schedulers without modifying their placement logic~(\S\ref{sec:design}). When jobs arrive, the scheduler assigns GPUs as usual. \sys monitors the resulting placement and checks whether any rack's fragmentation exceeds a threshold derived from its uplink count. If so, the controller invokes a migration solver to compute the minimum migrations needed to restore a congestion-free state. Once fragmentation is bound, the controller computes routes using a \textit{Perfect routing} scheme we designed~(\S\ref{sec:perfect}), that assigns each DP flow a dedicated uplink, achieving full path isolation. Daemons on each server execute migrations and configure local routing rules. By exploiting ML traffic sparsity and the self-reinforcing nature of low fragmentation, \sys keeps migration overhead to a minimum, typically requiring fewer than two moves per job arrival.

The core algorithmic contribution is the migration solver. We formulate the problem as an integer linear program (ILP): given a current placement, find a new placement that meets the fragmentation target while minimizing worker movements~(\S\ref{sec:ring}). We first solve this for workloads with a single ring-based collective and prove a tight bound: any placement can be defragmented into at most 2 cross-rack fragments per ToR~(\S\ref{sec:fragmentation}). This bound is both achievable and necessary in the worst case. Using insights from the proof, we extend the formulation to hybrid parallelism, where each server runs multiple workers forming parallel communication rings~(\S\ref{sec:model_parallelism}).

Realizing this design requires addressing two challenges. The first is ILP scalability: integer programs have worst-case exponential complexity. However, our analysis shows that the ILP solve time scales with the number of moves required~(\S\ref{sec:ring}). Combined with the abundance and self-reinforcing nature of low-fragmentation cluster states, the number of moves required is two or less 80\% of the time. This keeps solve times low in practice. 
For a cluster with $1,024$ GPUs, solving the ILP takes $0.89$ seconds on average, triggered on average every $102$ minutes at 80\% load~(\S\ref{sec:ilp_results}).

The second challenge is migration overhead. Live migration is considered complex and disruptive, but modern DNN training makes it surprisingly simple. Training frameworks already provide checkpointing to store and restore state for fault tolerance, and we build migration directly on this infrastructure. We implement in-memory checkpoint-and-restore over RDMA using PyTorch's Distributed Checkpointing library~\cite{pytorch_dcp}. 
 
To evaluate \name, we build a custom Rust-based simulator for modeling multi-tenant distributed ML training, modeling clusters of up to $2,048$ H200 GPUs and a $400$ Gbps network. We implement several models, including variants of GPT-3~\cite{gpt3, megatron}, GPT-OSS~\cite{gpt_oss_20b,gpt_oss_120b}, Llama2~\cite{llama2}, and Llama3~\cite{llama3}~(\S\ref{sec:sim_setup}). We also prototype \name in a testbed with five nodes, each equipped with one $A100$ GPU and a 100~Gbps RDMA NIC, and perform an end-to-end migration benchmark. In simulation, we demonstrate that on a cluster of $1,024$ GPUs \name improves average job completion time by $14$\% over the next best baseline~(\S\ref{sec:eval_job_isolation}). On a cluster of $2,048$ GPUs with an oversubscription ratio of 16:1, we keep the p99 job-completion slowdown to just 5\%~(\S\ref{sec:eval_slo}). Our testbed prototypes support efficient migration, with $9.02$~s of system overhead from the start of migration to resuming training~(\S\ref{sec:prototype}). To the best of our knowledge, this is the first system to mitigate congestion between competing ML training jobs using migration as the primary remedy. The MonkeyTree codebase is open source~\cite{monkeytree_repo}.

This work does not raise any ethical issues.

\section{Multi-Tenant Training and Congestion} \label{sec:motivation}
\begin{figure}[t]
  \centering
  \includegraphics[width=.8\columnwidth]{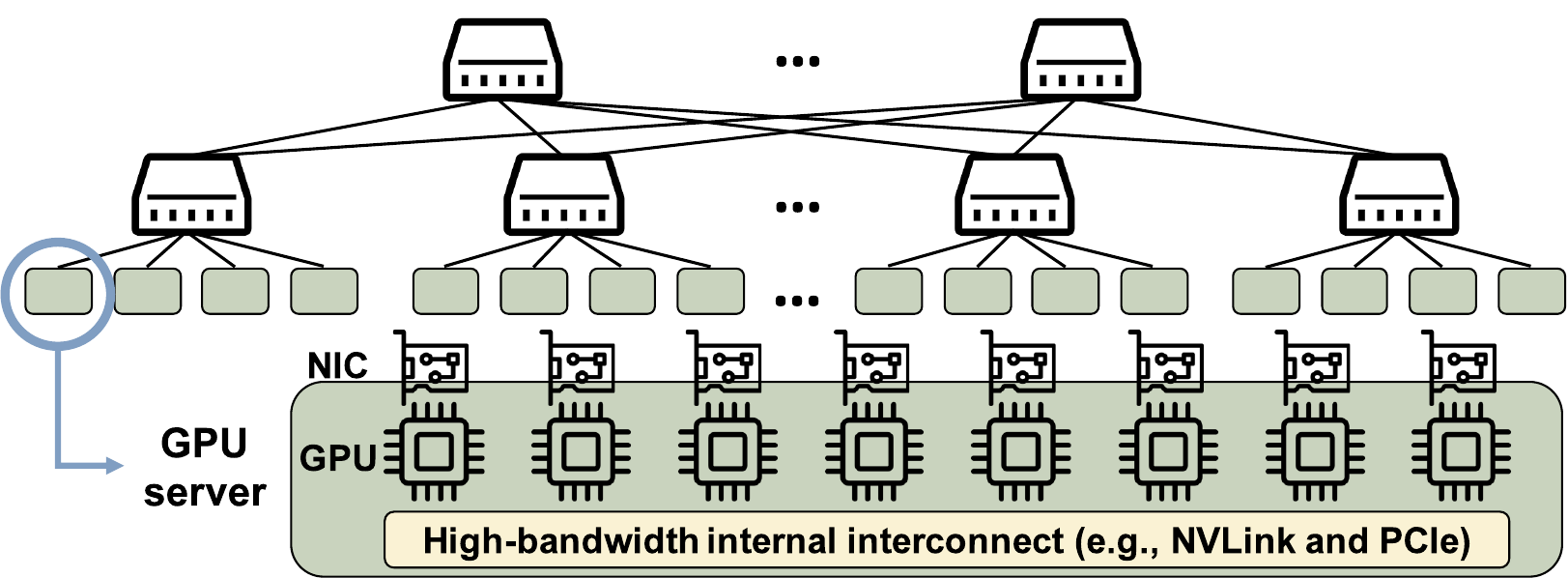}
  \caption{A multitenant GPU cluster. }
  \label{fig:mt_cluster}
\end{figure}

 \begin{figure*}[t]
   \centering
    \captionsetup[subfigure]{aboveskip=-1pt,belowskip=-1pt}
    \begin{minipage}{0.62\textwidth}
    \centering
    \includegraphics[width=.88\columnwidth]{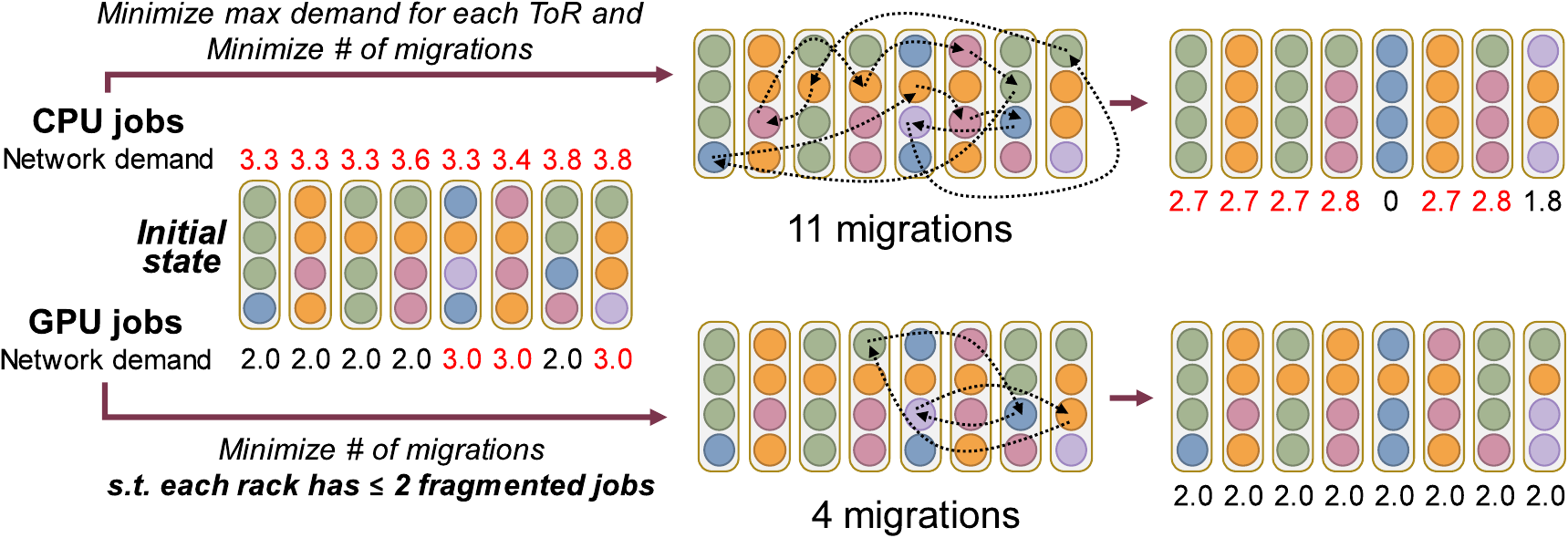}
    \caption{Migration and network congestion in shared CPU vs. GPU clusters. \normalfont Different colors represent workers in different jobs; each ToR switch supports two units of demand. The CPU cluster (top) requires 11 migrations, yet inter-job congestion persists on racks 4 and 7. The GPU cluster (bottom) needs only 4 migrations to eliminate congestion.}
    \label{fig:mig_example}
  \end{minipage}
  \hfill
  \begin{minipage}{0.35\textwidth}
  \centering
  \includegraphics[width=.9\columnwidth]{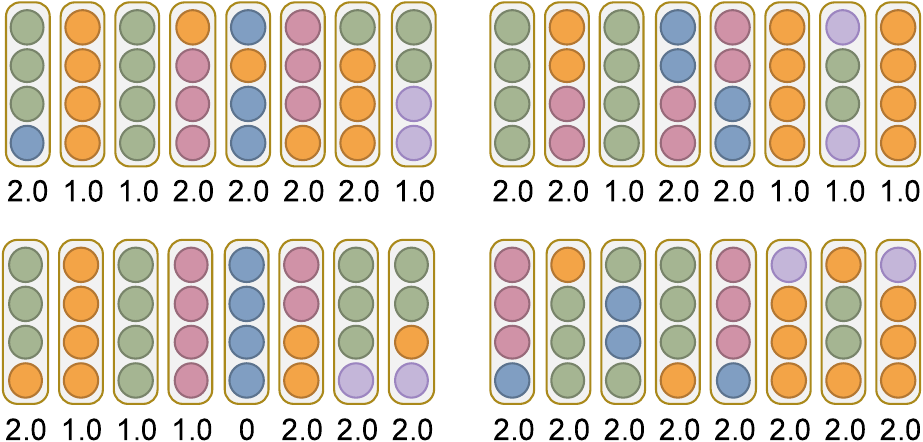}
  \caption{Sample congestion-free states for GPU clusters, for jobs shown in Figure~\ref{fig:mig_example}. \normalfont Despite different placements, all states have two or fewer fragmented jobs, thereby generating at most two units of demand per rack.}
  \label{fig:mig_ubi}
\end{minipage}  
\vspace*{-3mm}
\end{figure*}

Multi-tenant training is ubiquitous in today's cloud environment~\cite{crux, sglb, harmonics, cassini, mltcp, foresight}. A typical GPU cluster consists of GPU servers interconnected by a hierarchical network, shown in Figure~\ref{fig:mt_cluster}. Each GPU server contains multiple GPUs, each with a dedicated NIC, as well as an internal high-bandwidth interconnect shown at the bottom of Figure~\ref{fig:mt_cluster}. In these shared GPU clusters, each training job occupies one or more GPU servers, making the network the primary shared resource. Unfortunately, network congestion significantly degrades training performance, with Alibaba reporting up to 36.1\% of the jobs experience throughput reduction due to congestion~\cite{crux}.

\subsection{Prior Solutions}
Two directions have emerged to address this. The first invests in network infrastructure: packet spraying over full-bisection bandwidth topologies delivers near-ideal performance but requires expensive specialized hardware such as Nvidia's SuperNICs~\cite{supernic, spectrum_x}. Many deployments, therefore, accept oversubscribed networks to reduce cost~\cite{meta_rdma,ncclx,hpn}. The second direction employs software techniques, such as routing~\cite{crux, sglb} and flow scheduling~\cite{cassini, mltcp}, to mitigate congestion in oversubscribed networks. We evaluate Crux as an example of these approaches, on a cluster with $1,024$ GPUs and 4:1 network oversubscription ratio. We observe a $26\%$ average slowdown despite the mitigation: when aggregate traffic exceeds network capacity, flow collisions become inevitable regardless of strategy.  Flow scheduling approaches like Cassini attempt to overlap one job's communication with another's computation, but this strategy is fragile because it requires well-matched job characteristics and becomes unstable when iteration times differ by even a slight amount~\cite{mltcp_full}.

Both approaches overlook a more fundamental solution: reducing traffic entering the network in the first place. A job whose workers reside entirely within a single rack generates no traffic beyond the ToR layer, whereas a \textit{fragmented} job spanning multiple racks requires communication through the spine. When too many fragmented jobs share a rack, the aggregate uplink demand exceeds the ToR's uplink capacity, leading to congestion. Job scheduling doesn't resolve fragmentation: our evaluation shows that in a cluster of $1,024$ GPUs across $16$ racks, up to $93$\% of ToRs carry more uplink demand than their capacity even under a locality-optimized scheduler. This happens in production systems at scale as well. We analyzed Alibaba's report and found that almost 50\% of their jobs are fragmented~\cite{crux}. Fragmentation is not merely a heuristic artifact: cluster placement is an online multidimensional bin-packing problem with unknown future arrivals. Offline placement is NP-hard, and online lower bounds rule out optimality. Thus, fragmentation is inevitable even with a well-designed scheduler~\cite{bin_packing}. The natural remedy is \textit{defragmentation} through job \textit{migration}, yet prior work on network contention in GPU clusters ignores it, a legacy of CPU-centric multitenant clusters where migration was impractical. 

\subsection{Migration Challenges in CPU Workloads} \label{sec:cpu}
Migration-based approaches for CPU workloads fall into two categories.
The first addresses resource sharing at the server level, where CPU cores and memory are the primary contended resources~\cite{nsdi_live_migration_core_1, migration_core_2, sandpiper}. Early GPU migration works, such as Gandiva~\cite{gandiva}, adopted the same granularity, splitting GPUs across jobs on a server. This does not apply to our setting, where the network above between ToR and spine switches is the primary contended resource~\cite{crux, meta_imc}. 
\looseness=-1

The second jointly optimizes placement and routing to reduce network contention~\cite{joint_migration_routing}. However, CPU workloads generate arbitrary traffic patterns, yielding a multi-objective problem: minimize network traffic while minimizing migrations. This formulation has fundamental limitations. Consider five jobs (color-coded) on 32 servers across eight racks in Figure~\ref{fig:mig_example}. Each ToR has two uplinks supporting two units of demand, and each node sends one unit of traffic uniformly to others in its job. The initial placement is fragmented, with racks demanding up to 3.8 units (Figure~\ref{fig:mig_example}, top-left). The placement that minimizes peak per-ToR traffic requires at least 11 migrations (Figure~\ref{fig:mig_example}, top), yet congestion persists: only two racks stay within uplink capacity, and racks 4 and 7 still exhibit inter-job contention between the green and pink jobs. \looseness=-1

This illustrates two problems. First, eliminating congestion is impossible: no placement achieves zero inter-job contention in this example, and alternative solutions simply shift congestion to different jobs. Without a policy specifying job priorities, there is no clear optimum. Second, good states are sparse: diverse CPU traffic patterns make low-congestion placements rare targets, requiring one-third of the cluster to migrate in this example alone—a ratio that worsens with scale. \looseness=-1

In summary, migration for CPU workloads faces fundamental obstacles. Does this change for GPU workloads?
\enlargethispage{0.5\baselineskip}
\section{Tractability of GPU Migration} \label{sec:gpu_mig}
Distributed training jobs consist of computation and communication operations. Communication arises from parallelizing model training across GPUs, creating data dependencies that manifest as network transfers. Modern LLM training combines multiple parallelization strategies, each generating distinct traffic patterns.

Table~\ref{table:traffic_size} shows traffic generated by one iteration of Llama3-70B and GPT3-175B on 128 GPUs with a local batch size of 16. Tensor parallel (TP) traffic is the largest but remains confined to high-bandwidth interconnects such as NVLink within a server. Among traffic entering the datacenter network, data-parallel (DP) collectives dominate: 7.5$\times$ and 29.4$\times$ larger than pipeline-parallel (PP) traffic for Llama and GPT, respectively. Congestion, therefore, manifests most significantly when DP collectives from different jobs collide. Hence, we define a placement to be a \textit{congestion-free state} if no DP traffic across jobs competes in the network. To minimize network congestion, we focus on managing DP traffic.

DP replicates the model across GPUs, with each GPU processing a different data batch. Fully Sharded Data Parallelism (FSDP) extends DP by sharding parameters across GPUs. Both require synchronizing gradients or parameters at the end of each training iteration through collective communication. These collectives generate traffic proportional to model size with the same volume and pattern, so we treat them identically and refer to both as DP traffic. When combined with other parallelisms such as PP and TP, only a subset of GPUs shares the same model replica; we call this subset a \textit{replica group}, and each performs its own collective independently.

\para{Ring-based collectives.} We focus on ring-based collective algorithms, which are bandwidth-optimal and generate the sparsest traffic pattern~\cite{ncclx}. When a replica group spans multiple racks, each ToR generates one flow entering and leaving, regardless of how many workers reside in that rack or their ordering. We refer to such traffic as DP-flows throughout the rest of the paper. This sparsity fundamentally changes the property of congestion-free states in shared GPU clusters. \looseness=-1

\begin{table}[t]
    \centering
    \scriptsize
    \setlength{\tabcolsep}{3.5pt}  
    \begin{tabular}{@{}c|cccc|cccc@{}}
    \toprule
    \textbf{Model}       & \multicolumn{4}{c|}{\textbf{Llama3-70B}}                           & \multicolumn{4}{c}{\textbf{GPT3-175B}}                          \\ \midrule
    \textbf{Parallelism} & \textbf{Deg.} & \textbf{Total} & \textbf{\% of DP} & \textbf{Net?} & \textbf{Deg.} & \textbf{Total}  & \textbf{\% of DP} & \textbf{Net?} \\ \midrule
    \textbf{TP}          &       8       &  35.8~TB       &  426\%           & N             & 8              & 68.8~TB         &  3276\%        & N             \\
    \textbf{DP/FSDP}     &       4       &  840~GB        &  100\%           & Y             & 4              & 2.1~TB          &  100\%         & Y             \\
    \textbf{PP}          &       4       &  48~GB         &  5.7\%           & Y             & 4              & 72~GB           &  3.4\%         & Y                \\ \bottomrule
    \end{tabular}
    \caption{Traffic composition of common LLMs with different parallelism and whether the traffic is carried in the network.}
    \label{table:traffic_size}
\end{table} 

\subsection{Congestion-Free States are Achievable}
Consider again the example in Figure~\ref{fig:mig_example}, now with GPU jobs. We assume each server has one GPU and runs training with DP-only (i.e., one replica group per job); we discuss the multi-GPU and hybrid parallelism cases in Section~\ref{sec:model_parallelism}. 
Since DP traffic dominates network congestion, we analyze only DP synchronization traffic, where each GPU sends one unit of demand to its neighboring node. The total network demand generated by each job is the same across CPU and GPU jobs; however, unlike the CPU case, the traffic matrix is \textit{sparse}: ring-based collectives generate exactly one flow per ToR per fragmented job. As long as each ToR contains at most two fragmented jobs, the number of outgoing flows stays below the number of uplinks.

This sparsity resolves the ill-defined optimization problem from Section~\ref{sec:cpu}. In particular, there exist states in which the traffic leaving each ToR does not exceed its uplink capacity (bottom-left of Figure~\ref{fig:mig_example}). We define such states as \textit{congestion-free}, since this condition allows each DP flow to be assigned a dedicated uplink, achieving \textit{path isolation}. Under minimal uplink provisioning, such states always exist (\S\ref{sec:fragmentation}). For CPU workloads, operators must balance network contention against migration overhead~\cite{joint_migration_routing, joint_migration_routing_2, joint_placement_routing_3}. In contrast, for ML workloads this tradeoff disappears: because congestion-free states are always attainable, they can be enforced as a constraint, allowing the optimization to focus solely on minimizing migrations.

Now, we turn to minimizing migrations: given a contended state, what is the minimum number of worker movements to reach any congestion-free state? The bottom of Figure~\ref{fig:mig_example} shows such placement with the minimum amount of migration, requiring only 4 compared to 11 in the CPU case. Such a reduction stems from another property of GPU jobs: congestion-free placements are readily available.

\subsection{Congestion-Free States Are Abundant}
The achievable nature of congestion-free states implies another useful property; \textit{any} state where the number of flows leaving each ToR does not exceed its uplink count is equally good, and we do not need to find a single optimal arrangement. For instance, Figure~\ref{fig:mig_ubi} shows four other congestion-free placements for GPU workloads, in addition to the bottom right placement in Figure~\ref{fig:mig_example}. Notice that even seemingly fragmented placements, like the bottom-right one, are congestion-free.

This flexibility compounds in several ways. First, workers within a replica group are order-insensitive: they form a ring in any order through a simple re-assignment of ranks for communication, and still produce the same sparse traffic pattern in the network. Second, the constraint is local to each ToR, as satisfying the uplink bound at one ToR is independent of other ToRs. Consequently, unlike CPU workloads, where low-congestion placements are rare and expensive to reach, congestion-free states in a shared GPU cluster are plentiful. Hence, we can always reach one from any contended state through a small number of migrations (§\ref{sec:eval}).

\subsection{Congestion-Free States are Self-Reinforcing}
Once a cluster reaches a congestion-free state, maintaining it requires only incremental effort. A single new job spans a limited number of ToRs, so its arrival can violate the uplink constraint at only those few racks—a small, localized perturbation to an otherwise valid configuration. Because valid placements are abundant, restoring compliance never requires a far-reaching search: a small number of migrations among the affected racks suffices. This creates a virtuous cycle: few migrations mean low overhead, low overhead makes proactive maintenance practical, and a well-maintained cluster presents only simple, local violations to correct. Rather than gradually drifting into a highly contended state that demands expensive cluster-wide reorganization, the system naturally remains in the neighborhood of valid configurations.

\section{\name System Design}\label{sec:design}
\begin{figure}[t]
    \centering
    \includegraphics[width=.88\linewidth]{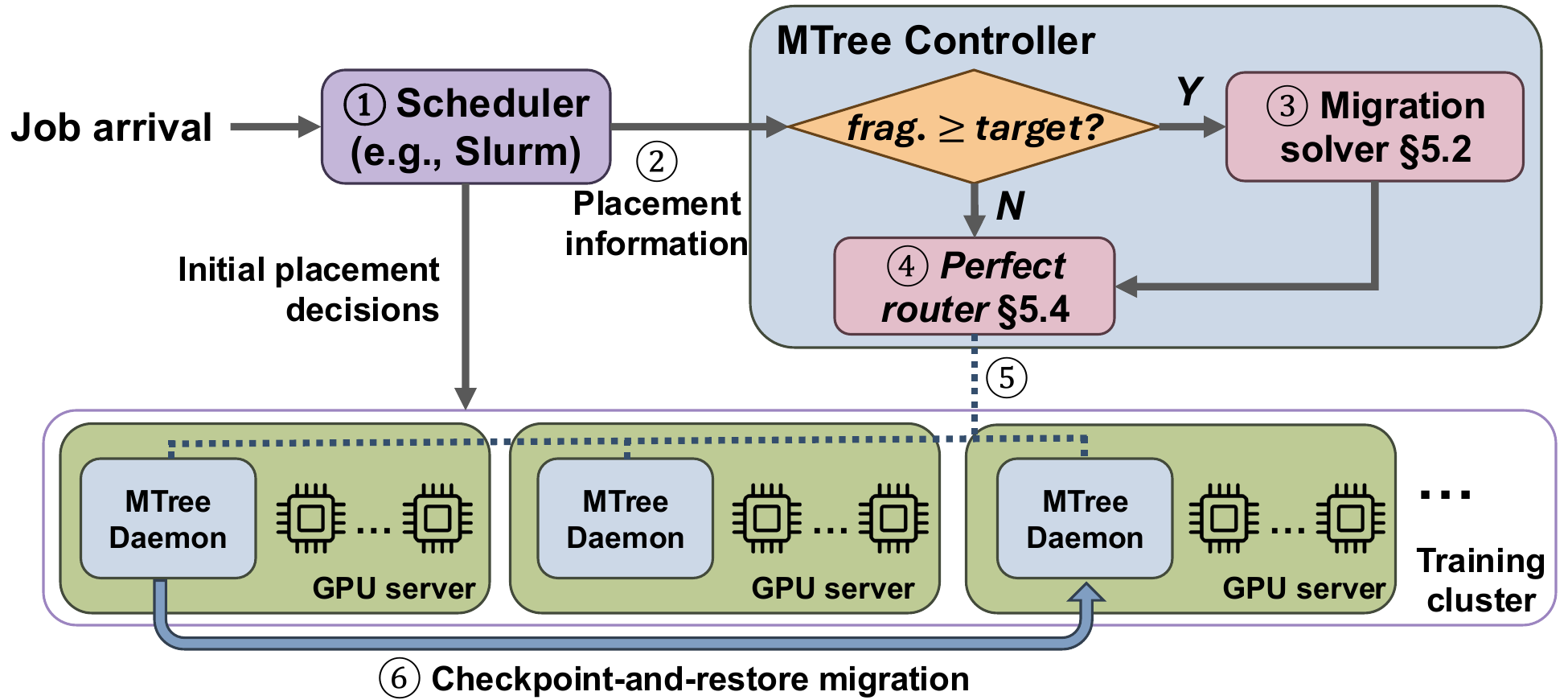}
    \caption{Overview of \name. {\normalfont A centralized controller and per-node daemons operate alongside existing schedulers.}}
    \label{fig:system_design}
\end{figure}
Our analysis in Section~\ref{sec:gpu_mig} shows that migration in shared GPU clusters is both practical and desirable. We introduce \name, a system that eliminates network contention in shared GPU clusters through migration-based defragmentation.

\para{Overview.} Figure~\ref{fig:system_design} illustrates \name's architecture and workflow. The system consists of a centralized controller (\S\ref{sec:controller}) and distributed daemons (\S\ref{sec:daemons}), operating alongside existing job schedulers without modifying their placement logic. \looseness=-1

The workflow begins when a job arrives, and the existing job scheduler \circled{1} assigns GPU placements to the new job. Rather than intervening in placement decisions, \name monitors the resulting cluster state \circled{2} and initiates migration only when necessary. We set a fragmentation threshold based on the ToR uplink count to ensure path isolation is always achievable (\S\ref{sec:fragmentation}). If any rack exceeds this threshold, the controller invokes the migration solver \circled{3} to compute the minimum worker movements needed to restore a congestion-free state (\S\ref{sec:ring}). Whether or not migrations are needed, the controller computes routes \circled{4} that provide path isolation for all DP flows (\S\ref{sec:perfect}). It then sends migration instructions and routing assignments to the daemons \circled{5}. Each daemon executes its assigned migrations via an efficient checkpoint-and-restore operation \circled{6} and configures local routing rules to steer DP traffic onto isolated paths.

\subsection{\name Controller} \label{sec:controller}
The controller interfaces with the cluster's job scheduler (e.g., Slurm~\cite{slurm}) to observe placement decisions. It does not modify the scheduler's allocation logic, but allows the scheduler to make an initial decision after job arrival. If the placement decision violates the fragmentation constraint, \sys intervenes and computes a new placement that satisfies the constraint while minimizing migrations.\looseness=-1

\para{Migration decisions.} When a new job's placement causes any rack's fragmentation to exceed the threshold, the controller computes the minimum migrations to restore compliance. The threshold is set to equal the number of uplinks per ToR, ensuring that path isolation is always achievable. Section~\ref{sec:algorithm} details the migration algorithm.

\para{\textit{Perfect routing }solver.} After each job arrival or migration, the controller reconfigures routing to achieve path isolation. \name uses a deterministic routing scheme, called \textit{Perfect routing}, that assigns each DP traffic a dedicated uplink, guaranteed to exist as long as fragmentation stays within the threshold. The controller transmits routing assignments to the daemons, which compute ECMP hashes that steer traffic onto the assigned paths, similar to prior work~\cite{crux, pecmp}.

\subsection{\name Daemons}\label{sec:daemons}
Daemons run on each server and handle job lifecycle operations: launching jobs with pre-computed ring orderings, executing migrations, and updating routing configurations.

\para{Migration execution.} We implement migration itself as checkpoint-and-restore over GPUDirect RDMA. When the controller initiates a migration, the source daemon pauses training after the current iteration completes and checkpoints model state to the destination. The destination daemon receives the checkpoint into CPU memory, then loads it onto GPUs and resumes training with the replica group. This approach leverages existing checkpointing infrastructure and avoids overwhelming GPU memory during transfers. Section~\ref{sec:prototype} describes our prototype implementation.

\section{\sys Formulation} \label{sec:algorithm}
In this section, we first show that for a GPU cluster with at least two uplinks per ToR, congestion-free states are well-defined and always achievable (\S\ref{sec:fragmentation}). We then formulate the problem of minimizing the number of migrations to bound the number of DP-flows that enter the fabric as an Integer Linear Program (ILP) (\S\ref{sec:ring}, \S\ref{sec:model_parallelism}). Finally, we present a routing scheme, \textit{\textit{Perfect routing}}, that exploits the bound on DP-flows to guarantee path isolation, nearly eliminating congestion (\S\ref{sec:perfect}).

\subsection{Bounding Fragmentation Degree}
\label{sec:fragmentation}

\begin{figure}[t]
    \centering
    \includegraphics[width=0.88\linewidth]{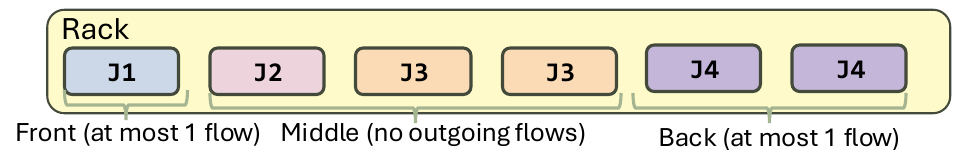}
    \vspace*{-1mm}
    \caption{Illustration of bounding fragmentation.}
    \label{fig:bounding_fragmentation}
    \vspace*{-1mm}
\end{figure}

Formally, we define a rack's \textit{fragmentation degree} as the number of DP flows originating from this rack that have destinations outside. For pure DP jobs, this is the number of fragmented jobs under the rack. We begin by characterizing what is achievable for a given cluster and set of jobs. \looseness=-1

\begin{theorem}
For any cluster and any set of pure DP or FSDP jobs running on it, there exists a job placement such that no rack has a fragmentation degree greater than two.
\end{theorem}

\textit{Proof}. Order the jobs arbitrarily, and place their workers sequentially across the cluster from left to right. For any rack, this placement induces at most three contiguous regions. The front consists of the trailing workers of a job that began on a previous rack and finishes on the current one. The end consists of the leading workers of a job, with the remaining workers placed on subsequent racks. The middle consists of jobs that are entirely contained within the rack.
Jobs in the middle region do not contribute to fragmentation, since all of their communication remains local to the rack. Moreover, the front and end regions each contain at most one job that spans racks. Therefore, each rack can contain at most two fragmented jobs, implying that the fragmentation degree of any rack is at most two. Figure~\ref{fig:bounding_fragmentation} provides a visualization of this argument. \looseness=-1

We note that there are scenarios in which achieving a fragmentation degree better than two is impossible. For example, consider two jobs whose sizes each exceed the capacity of a single rack, running on a cluster with three racks. In this case, the two jobs must overlap on at least one rack, forcing that rack to have a fragmentation degree of at least two. \looseness=-1

\begin{table}[t]
\centering
\scriptsize
\setlength{\tabcolsep}{3.5pt}  
\begin{tabular}{@{}cll@{}}
\toprule
\multirow{5}{*}{\textbf{Input}} & $w_0(s, t)$              & Initial placement: \# of job $s$ workers on ToR $t$            \\
                       & $\lambda$                & Fragmentation threshold                                        \\
                       & $S=\{s_0, \cdots, s_j\}$ & Set of jobs to be placed, where job $j$ requires $s_j$ workers \\
                       & $T$                        & Number of ToRs (racks) \\
                       & $C$                        & ToR capacity, i.e., number of nodes below a ToR                 \\ \midrule
\textbf{Output}                 & $w(s, t)$                & Final placement: \# of job $s$ workers on ToR $t$              \\ \bottomrule
\end{tabular}
\caption{Variable definitions.}
\label{tab:var_def}
\end{table}
\begin{listing}[t]
\centering
{\footnotesize

\begin{align}
\min &\textstyle\sum_{\text{all }s, t} \frac{1}{2}\left|w(s, t) - w_0(s, t)\right| \label{eq:obj} \\
\text{s.t.} \quad
& \textstyle\sum_{t<T}w(s, t) = s_j && \forall\ \text{Job } j\label{eq:job} \\
& \textstyle\sum_{s<|S|}w(s, t) \le C  && \forall\ \text{ToR }t \label{eq:tor} \\
& \left|\left\{s : w(s, t) > 0\ \text{and} \textstyle\sum_{t' \neq t} w(s, t') > 0\right\}\right| \leq \lambda && \forall\ \text{ToR }t \label{eq:frag} \\
&w(s, t) \ge 0 && \forall\ s, t \label{eq:positive}
\end{align}
}
\caption{\sys's optimization for minimizing migration.}
\label{lst:lp-obj}
\end{listing}

Once a placement limits the number of DP-flows exiting each rack to fewer than the number of ToR uplinks, a congestion-free cluster requires only a \textit{routing scheme} that assigns each flow to a dedicated uplink, thereby guaranteeing path isolation. We derive such a routing scheme, which we call \textit{\textit{Perfect routing}}, for a two-tier network topology, detailed in Section~\ref{sec:perfect}. Combined with the worst-case minimum fragmentation degree of two, our routing scheme implies that provisioning just two uplinks per rack is sufficient for a congestion-free network for shared GPU clusters.

The above arguments establish that a congestion-free state is \textit{achievable}. Next, we formulate \sys's core algorithm that finds a sufficiently defragmented placement from an initial fragmented one, minimizing the number of migrations between the two. We assume each rack has more than two allocated uplinks and aim to find a congestion-free state in which no rack's fragmentation degree exceeds the number of uplinks. \looseness=-1

\subsection{\sys ILP for DP and FSDP Jobs}
\label{sec:ring}
In this section, we identify the minimum number of migrations needed to defragment every rack below a set threshold, thereby achieving a nearly congestion-free state. We focus on pure DP and FSDP jobs here, and extend our formulation to model parallelism in Section~\ref{sec:model_parallelism}.

Table~\ref{tab:var_def} defines the variables we use in this formulation, and Listing~\ref{lst:lp-obj} illustrates \sys's optimization problem. For the complete casting as an ILP, see Appendix~\ref{appendix:full_ilp}. We model this problem as follows: let $S$ denote the set of jobs, and let $T$ denote the number of racks. We construct a bipartite graph in which the left nodes $s \in S$ correspond to fragmented jobs in the cluster, and the right nodes $t$ ($0 \leq t < T$) correspond to the racks. For every pair $(s, t)$, we introduce an edge whose weight $w(s, t)$ denotes the number of workers of job $s$ currently placed on rack $t$. Let $w_0(s, t)$ denote the initial worker placement. The resulting optimization objective is thus Equation~\ref{eq:obj} in Listing~\ref{lst:lp-obj}, capturing the total number of worker migrations as the $\ell_1$ distance between the initial and final assignments.

For the constraints, Equations~\ref{eq:job} ensure that each job is fully assigned, while Equation~\ref{eq:tor} ensures that each rack receives fewer jobs than its capacity. The key constraint is Equation~\ref{eq:frag}: the left-hand side constructs a set of fragmented jobs on each rack $t$. The first part of the set construction checks whether rack $t$ contains job $s$, while the second checks whether job $s$ is fragmented. Hence, the cardinality of this set is the number of fragmented jobs, which were required to be at most $\lambda$ (RHS).  Satisfying this constraint guarantees that the resulting placement is sufficiently defragmented to meet the routing requirements for a congestion-free state.

This optimization problem is a variant of the sparsity-constrained transport problem, which is NP-hard~\cite{sparsity_constrained_transport}. In practice, we transform it into an integer linear programming (ILP), which surprisingly has a practical solve time. For example, for clusters with up to 1,024 GPUs, solve times are on average under one second (Section~\ref{sec:ilp_results}). 

We attribute this efficiency to the ability of branch-and-bound solvers~\cite{branch_bound} to exploit the abundance of the solution space. To build intuition for the runtime, let $\gamma$ denote the number of migrations in searching for a solution. The number of distinct cluster states reachable within $\gamma$ moves is $O(N^{2\gamma})$, where $N$ is the cluster size. As branch-and-bound progressively discovers feasible solutions with increasing migration counts, large portions of the search space are pruned, since the solver discards solutions requiring the same or more moves. We believe that, because the optimal number of migrations is typically small (\S\ref{sec:ilp_results}), the solver explores a limited effective search space, resulting in fast ILP solving time.

\vspace*{-3mm}
\subsection{Extending to Model Parallel Jobs}
\label{sec:model_parallelism}
The prior formulation works for DP and FSDP jobs but fails to capture additional complexities introduced by model parallelisms. We address additional parallelisms in turn.

\para{Pipeline parallelism (PP).} Pipeline parallelism shards the model between layers into distinct stages that run on separate GPUs. Activations are communicated point-to-point between stages. Data-parallel replicas are thus across each stage, and the subsequent all-reduce (DP-flows) only happens within these stages. We only use our \textit{Perfect routing} algorithm to load balance pipeline traffic and ignore accommodating them in our formulation, because they are small (\S\ref{sec:gpu_mig}) and therefore have little impact on end-to-end job completion time. Instead, we update our formulation to treat pipeline stages as independent jobs, thereby controlling each stage's DP-flow separately. 

\para{Tensor parallelism (TP).} Tensor parallelism shards an individual layer of a DNN across GPUs, resulting in substantial communication volume. For this reason, the standard practice confines TP \textit{within} a single server, where GPUs communicate over high-bandwidth interconnects such as NVLink~\cite{nvlnvs, megatron}.

Each TP shard forms an independent DP replica group and produces its own all-reduce traffic. As a result, a fragmented job that previously generated a single ring per rack now produces as many DP-flows as there are GPUs per server when TP is employed. Following the standard configuration of 8 GPUs per server~\cite{megatron,llama3}, we assume a maximum TP degree of 8 throughout this paper, yielding up to 8 independent DP flows per fragmented job per rack. We now generalize our previous result on the attainable fragmentation degree for TP jobs. Let $d_s$ denote the TP degree of job $s$, and let $D_C = \max_{\{s\in S\}} d_s$ represent the maximum TP degree across a set of jobs. Then, the maximum number of DP-flows per server is $D_C$, and we have the following result:
\begin{corollary}
For any cluster with ML jobs, with maximum TP degree of $D_C$, there exists a job placement such that no rack has a fragmentation degree greater than $2D_C$.
\end{corollary}
The proof follows directly from the argument in Theorem~1, with each fragmented job now generating up to $D_C$ DP-flows. As a consequence, under standard tensor-parallel configurations with eight GPUs per server, avoiding congestion requires provisioning at least $16$ uplinks per rack.
We then refine the fragmentation constraint in Listing~\ref{lst:lp-obj}~(Equation ~\ref{eq:frag}) to weigh the contribution of each job $s$ by the number of rings it produces $d_s$. \looseness=-1

\para{Expert parallelism (EP).} In MoE models, each layer contains multiple expert sub-networks, only a subset of which activate for any given token. Expert parallelism (EP) places each expert on a different GPU. Between layers, an \textit{all-to-all collective} dispatches tokens to their assigned experts. Another all-to-all collective then combines the results, generating communication across all participating GPUs. Without replicating experts across GPUs, EP generates all-to-all traffic on par with DP, violating \sys's assumption where sparse DP traffic dominates. While defragmentation continues to benefit EP workloads, \name cannot guarantee congestion-free execution. We evaluate the impact of \name on EP-dominant workloads in \S\ref{eval:ep}, and discuss this limitation and potential extensions in \S\ref{sec:discussion}.

\subsection{Perfect Routing}
\label{sec:perfect}
In this section, we describe the \textit{Perfect routing} algorithm, which exploits our bound on the number of DP-flows entering the fabric to guarantee path isolation for them. For other mouse flows, such as pipeline-parallel traffic, \textit{Perfect routing} load-balances them via round-robin across available spines.

To route DP-flows, we model the problem as a bipartite edge-coloring instance. Each vertex on the left and right corresponds to a ToR switch, and for every flow between two distinct ToRs, we introduce an edge connecting the corresponding vertices. The objective is to assign a color to each edge, such that no two edges incident to the same vertex share a color. Here, each color corresponds to a distinct flow-to-uplink assignment. 

We compute such a coloring in polynomial time by reducing the problem to bipartite matching, and solve it using the Hopcroft–Karp algorithm~\cite{hopcroft_karp}. The number of colors required equals the maximum number of flows sent or received by any rack, which is always less than the number of uplinks in \sys. \looseness=-1

We note that when \textit{Perfect routing} executes separately on a fragmented cluster without migration, the number of required colors can exceed the number of physical uplinks per ToR. In these cases, we collapse multiple colors (flows) onto the available uplinks in a round-robin fashion to evenly distribute load, multiplexing multiple flows onto the same path. We evaluate \textit{Perfect routing} without \sys's defragmentation in Section~\ref{sec:eval} to isolate the performance impact of routing alone. 

\subsection{Beyond 2-tier Clos Topologies}
This section studies how to extend \name to other topologies.

\para{3-tier Clos topology.} A 3-tier Clos can be viewed as a set of 2-tier Clos \textit{pods} connected by an extra \textit{core-switch} tier. Our formulation already controls fragmentation across racks within each pod, but traffic can also congest the links between pods. We therefore also need to bound \textit{cross-pod fragmentation}. We introduce the following result:
\begin{corollary}
For a 3-tier Clos topology with ML jobs and maximum TP degree $D_C$, there exists a placement such that no pod has fragmentation degree greater than $2D_C$.
\end{corollary}
The proof follows the same argument as the rack-level bound, with pods taking the role of racks: under a left-to-right placement, each pod has at most two boundary-crossing jobs, and tensor parallelism scales each such job by at most $D_C$ DP flows. The main challenge is routing. In a 2-tier Clos, bounding rack fragmentation is enough: once each rack has no more DP flows than uplinks, \textit{Perfect-routing} can isolate them. In a 3-tier Clos, however, flows also share links from the second tier to the core tier, so rack-level isolation alone does not prevent collisions. Routing must therefore isolate flows both within pods and across pods.

We handle this by splitting routing into two coupled but separate steps. First, within each pod, we run \textit{Perfect-routing}. Flows that enter or leave the pod are treated as flows to or from a virtual endpoint at the pod boundary; this ensures that the intra-pod routing step accounts for the local uplinks consumed by inter-pod traffic, without yet determining the full inter-pod path. Second, we run \textit{Perfect-routing} between the second tier and the core tier for the inter-pod flows. This assigns each inter-pod flow a color, where each color corresponds to a distinct core-facing link. Because the second and core tiers are fully connected, each color can be realized by choosing a distinct core switch. As a result, flows of different colors use different core-facing links, preserving path isolation both across pods and within pods.

\para{Rail-optimized.} Rail-optimized topology design is becoming increasingly popular~\cite{rail_optimized,rail_only,hpn}. In each pod of a rail-optimized cluster, \textit{rail switches} instead of ToR switches directly connect all GPUs of the same rank across high-bandwidth domains. Each rail switch is, in turn, connected to the spine switches. Typically, the number of rail switches equals the high-bandwidth domain size and, therefore, the maximum TP size. For the rail-optimized topology, we have the following result:

\begin{corollary}
For a rail-optimized topology with ML jobs, there exists a placement and a rank assignment such that no rail requires more than $2$ uplinks to provide path isolation.
\end{corollary}

In other words, for a rail-optimized topology, the number of uplinks needed per rail switch can be independent of the TP degree as long as the rank assignment is rail-aligned. For TP segments, this is simple: all GPUs in the same TP group are placed on the same rail, so their traffic naturally separates. Pure DP and FSDP jobs are more subtle when servers are only partially occupied, because the rank assigned to each GPU determines which rail carries its outgoing DP flow. We leave the full treatment of this rank-assignment problem as future work.
\section{Evaluation} \label{sec:eval}

We evaluate \name in large-scale simulations and compare its performance with that of state-of-the-art algorithms for mitigating ML training congestion. In \S\ref{sec:sim_setup} we describe the setup and parameters for our experiments, as well as compared schemes. \S\ref{sec:eval_job_isolation} evaluates \name's ability to provide isolation between competing jobs and compares that to the performance of other algorithms. \S\ref{sec:eval_frag} also zooms in on a few specific traces to examine how fragmentation and network load evolve over time across various systems. \S\ref{sec:eval_slo} presents results on the impact of uplink provisioning on system performance. \S\ref{sec:eval_migration_moves} analyzes how many migrations \name performs and the ILP's solving time. In \S\ref{sec:migration_overhead}, we discuss the migration overhead cost to job performance. Finally, \S\ref{sec:eval_lambda} studies the impact of the migration threshold on the number of migrations. 
\subsection{Simulator Setup}
\label{sec:sim_setup}
\para{\name~Simulator.} To evaluate \name~on large-scale topologies, we build a custom Rust simulator for multi-tenant distributed ML training in \textasciitilde20k lines of code. For the network backend, we use a standard maxmin flow-level simulator~\cite{maxmin}. Our simulator represents individual training jobs as graphs of compute and communication steps, similar to prior work~\cite{simai, astrasim}. To get concrete values for the computation steps, we use the estimation techniques found in~\cite{rail_only}.

\begin{table}[t]
    \centering
    \scriptsize
    \setlength{\tabcolsep}{8pt}
    \begin{tabular}{@{}c|cccc@{}}
    \toprule
    \textbf{Model}       & \textbf{DP} & \textbf{FSDP} & \textbf{TP} & \textbf{PP} \\ \midrule
    \textbf{GPT3-13B}~\cite{megatron} & \checkmark & \xmark & \xmark & \xmark\\
    \textbf{GPT3-7B}~\cite{megatron} & \checkmark & \xmark & \xmark & \xmark\\
    \textbf{GPT-OSS-120B}~\cite{gpt_oss_120b} & \xmark & \checkmark & \xmark & \xmark \\
    \textbf{GPT-OSS-20B}~\cite{gpt_oss_20b} & \xmark & \checkmark & \xmark & \xmark \\
    \textbf{Llama2-70B}~\cite{llama2_70b} & \checkmark & \checkmark & \checkmark & \checkmark\\
    \textbf{Llama3-70B}~\cite{llama3_70b} & \checkmark & \checkmark & \checkmark & \checkmark\\
    \bottomrule
    \end{tabular}
    \caption{Jobs and parallelization strategies to generate traces.}
    \label{table:jobs}
\end{table}

\begin{figure*}[t]
    \centering
    \includegraphics[width=0.83\linewidth]{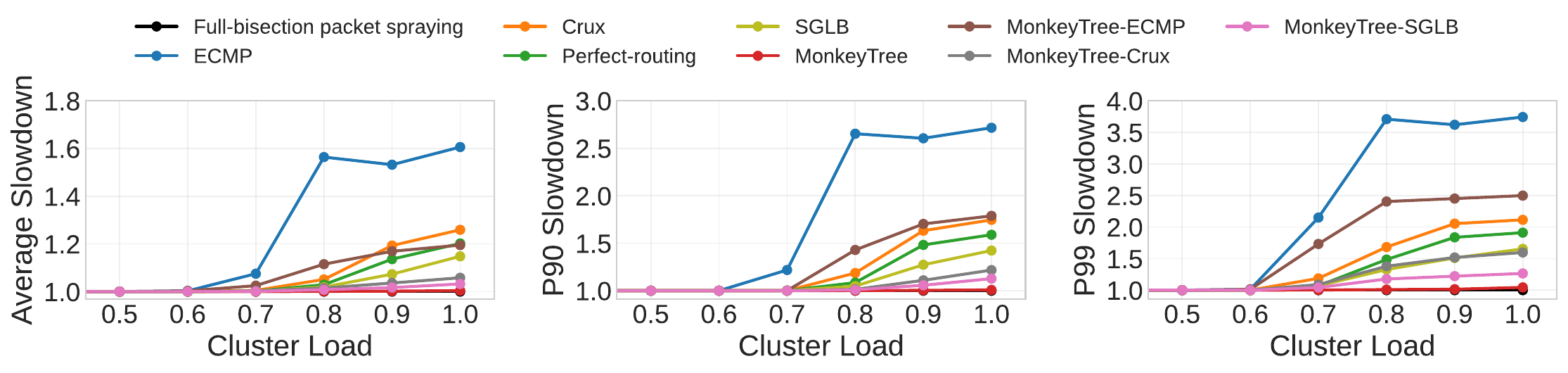}
    \vspace*{-1mm}
    \caption{Average, p90, and p99 job completion slowdown under varying cluster load for a 1,024-GPU cluster. \normalfont \name's migration-based approach reduces p99 slowdown by up to $3.59\times$ compared to routing-only schemes.}
    \label{fig:slowdown}
    \vspace*{-3mm}
\end{figure*}

\begin{figure*}[t]
    \centering
    \includegraphics[width=0.85\linewidth]{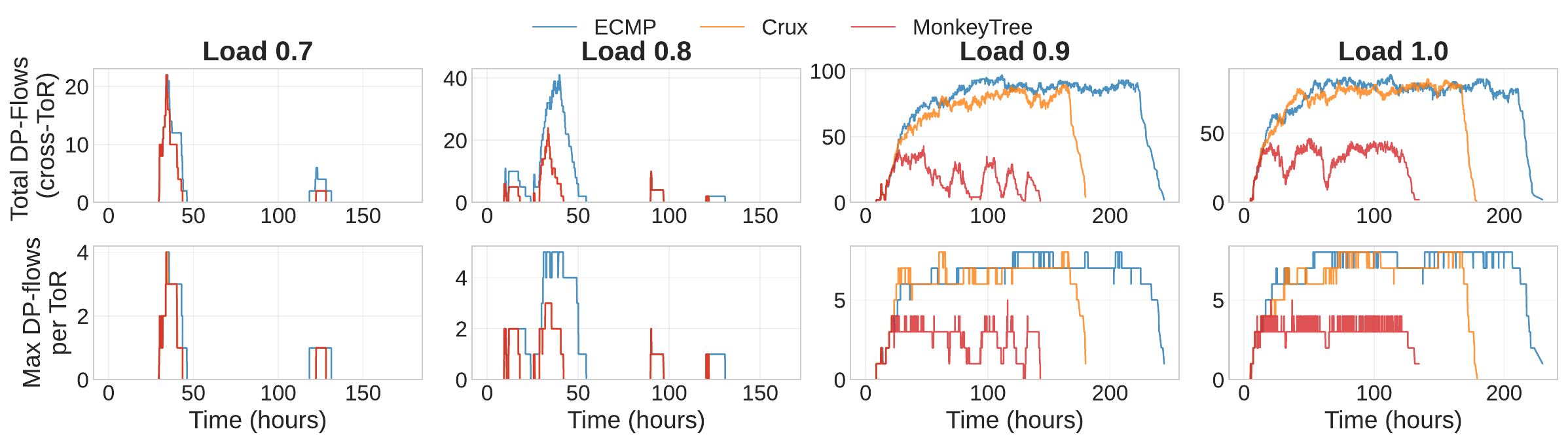}
    \caption{Cluster-wide and per-rack fragmentation over time at 70–100\% load in a 1,024-GPU cluster. \normalfont \name's active defragmentation keeps per-rack congestion within isolable bounds and yields up to $1.7\times$ faster job completion than routing-only schemes.}
    \label{fig:fragmentation}
    \vspace*{-3mm}
\end{figure*}

\para{Job Selection.} Table~\ref{table:jobs} presents the jobs we use in our simulation, along with the parallelization strategies available for selection. We generate job traces by randomly selecting jobs, determining which parallelization strategies, and specifying the degree of each parallelization.

\para{Traces.} We generate traces of job arrivals by modeling the system as a Poisson arrival process across a series of load levels from 50\% to 100\%. We plot the resulting slowdown for several systems in Figure~\ref{fig:slowdown}.

\para{Slowdown definition.} In several experiments, we plot the \textit{slowdown} in job completion. We define this slowdown as the division of the actual job runtime by its ideal completion time.

\para{Topology.} We use a standard two-tier spine tree topology, where hosts are connected to ToR switches, and these in turn are connected to a spine layer of switches. Unless otherwise specified, we place 8 hosts below each ToR, each equipped with 8 GPUs, and each GPU is given its own $400$~Gbps NIC. The oversubscription ratio changes across experiments.

\para{Cluster scheduler.} We use a locality-aware job scheduler similar to Slurm~\cite{slurm}. It first checks if any rack can fully accommodate the job it is scheduling. If it can, it will place it on the rack with the fewest available hosts. Otherwise, it finds the rack with the most available hosts, assigns as much of the job as possible to those hosts, and then recursively calls itself with the remaining workers to be scheduled.

\para{Migration simulation.} To simulate the overhead associated with migration, we implement a 4-step process. First, before migrating, jobs must synchronize at the iteration barrier. Secondly, to model the movement of the model weights, the migrating GPUs transfer their model shard as a flow from their GPU to the destination. Thirdly, migration does not complete until the destination GPU is free of any workers that were running on it. Finally, after all workers have reached their destination, we pause resumption for $10$ seconds to model re-initialization overhead, based on our testbed measurements (\S\ref{sec:prototype}). \looseness=-1

\para{Simulated algorithms.} We simulate several approaches to mitigate congestion in multi-tenant training clusters:
\begin{itemize}
\item \textbf{\name}. We run the \name~system as described in Figure~\ref{fig:system_design}. On every job arrival, \name controller monitors the fragmentation of each rack and intervenes if the set threshold is exceeded. In our experiments, we typically set the threshold to the number of uplinks per ToR. To study the benefit of migration alone, we also run variations of \name~built on top of ECMP, Crux, and SGLB rather than \textit{Perfect routing}.
\item \textbf{MonkeyTree3}. \name with the additional constraint for managing cross-pod fragmentation.
\item \textbf{ECMP}~\cite{ecmp}. ECMP is a classic load-balancing scheme that hashes the 5-tuple of flow information at each switch. This results in flows being pseudorandomly spread across the available paths.
\item \textbf{Crux}~\cite{crux}. Crux is a greedy scheduler that sorts jobs by expected GPU utilization and routes them sequentially to minimize overlap.
\item \textbf{\textit{Perfect-routing}-only} (\S\ref{sec:perfect}). We evaluate the \textit{Perfect routing} algorithm without \sys's migration as a separate baseline to isolate the performance impact of routing and migration.
\item \textbf{SGLB}~\cite{sglb}. SGLB is a load-balancer in which each switch dynamically routes flows away from congested paths. Switches exchange congestion information with one another to achieve rapid convergence.
\item \textbf{Full bisection with packet spraying}. Lastly, we simulate the cluster provisioned with a full-bisection fabric and packet spraying enabled, representing the performance upper bound at a high cost. 
\end{itemize}
\enlargethispage{0.3\baselineskip}
\begin{figure*}[t]
    \centering
    \includegraphics[width=0.85\linewidth]{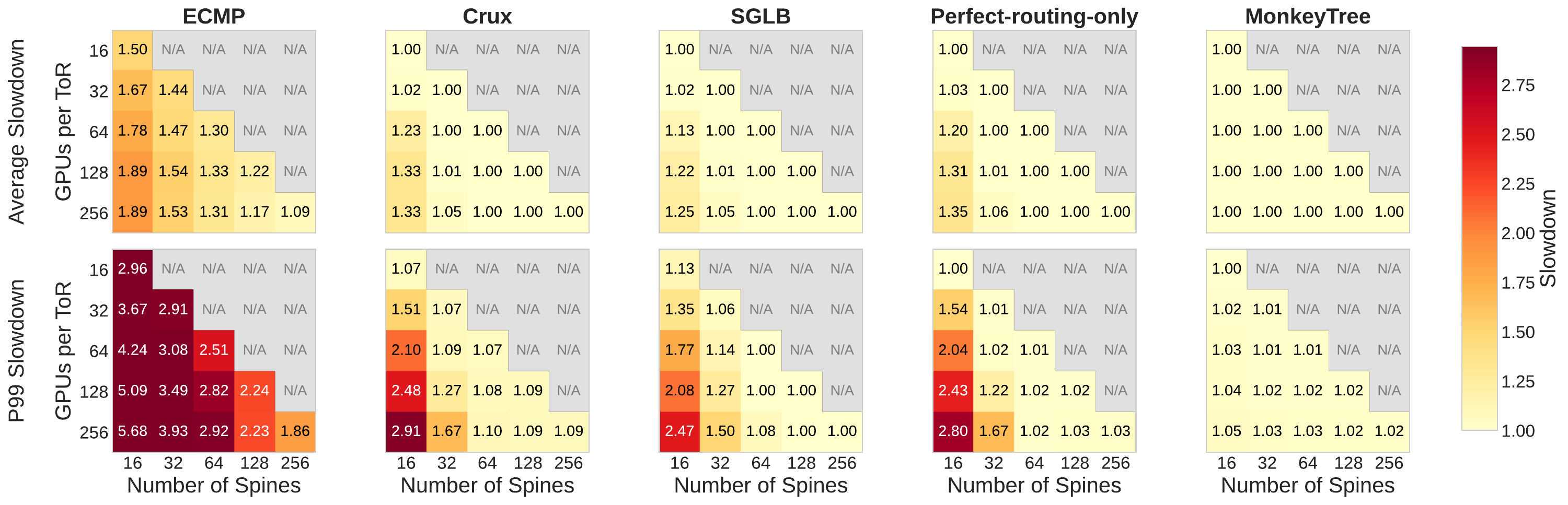}
    \vspace*{-1mm}
    \caption{Heatmaps of average and p99 slowdown across oversubscription ratios for a 2,048-GPU cluster for different systems. \normalfont \name achieves near-optimal performance at 16:1 oversubscription while competing schemes require up to $4\times$ more uplinks to match.}
    \label{fig:slo}
    \vspace*{-3mm}
\end{figure*}

\begin{figure*}[t]
    \centering
    \includegraphics[width=0.88\linewidth]{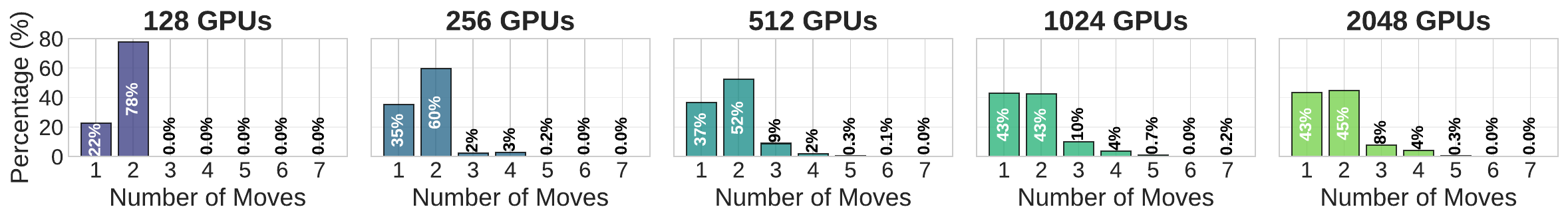}
    \caption{Distribution of job migrations required to restore low fragmentation across cluster sizes. \normalfont Over $80\%$ of defragmentation events require only one or two moves, and the average remains below 2 regardless of scale.}
    \label{fig:ilp_moves}
    \vspace*{-3mm}
\end{figure*}
\subsection{Performance Under Varying Cluster Load}
\label{sec:eval_job_isolation}
We begin by comparing the performance of various algorithms for running jobs in a shared cluster.
Figure~\ref{fig:slowdown} presents the average, p90, and p99 job completion slowdowns for job traces of varying load on the $x$-axis, with different algorithms. The $y$-axis shows the slowdown. The cluster comprises $1,024$ GPUs, with $64$ GPUs and $16$ uplinks per rack. 
\enlargethispage{0.1\baselineskip}

At 100\% load, \name~maintains a nearly contention-free state, having only a $4\%$ slowdown at the 99th percentile. On average, it performs within $0.5\%$ of packet-spraying. Slowdown stems only from migration overhead and brief pipeline-parallel transfer congestion.

For different routing algorithms, our \textit{Perfect-routing} algorithm guarantees path isolation when possible and load-balances otherwise. On the other hand, ECMP suffers from hash collisions, leading to elephant flows being routed together, while Crux's greedy heuristic occasionally fails to provide path isolation even when it is possible to do so. Notably, \textit{Perfect-routing}-only outperforms ECMP and Crux, beating Crux by up to 5\% on average. SGLB is the best of the non-\name baselines, as its dynamic rerouting enables finer-grained sharing of available links. \name beats all of these routing-only approaches by $1.58\times$ to $3.59\times$ in p99 job completion time. \looseness=-1

We demonstrate the benefit of migration by evaluating \name-variants with migration only (without the \textit{Perfect routing} algorithm). Notably, \name-ECMP outperforms Crux by $5$\% and matches \textit{Perfect-routing} at the highest load. At high load, fragmentation generates significant network contention that no routing scheme, however sophisticated, can fully resolve. \name addresses the problem at its root: by migrating jobs to reduce fragmentation, even a simple scheme like ECMP is sufficient to outperform more intelligent routing-only approaches. Compared to their corresponding baselines, \name-Crux improves slowdown by up to 32\% over Crux, while \name-ECMP attains up to a $1.5\times$ speedup over ECMP. \name-SGLB improves SGLB's p99 slowdown by $1.31\times$ at 100\% load. While SGLB outperforms \textit{Perfect-routing} on its own, with \name, \textit{Perfect-routing}'s performance isolation guarantees keep slowdown generally low. In contrast, SGLB may occasionally fail to converge to such a state, even when one exists. \looseness=-1

\label{sec:eval_frag}

Figure~\ref{fig:fragmentation} plots the cluster-wide and per-rack maximum fragmentation over time for ECMP, Crux, and \name on a 1,024-GPU cluster with three uplinks per rack, using pure DP/FSDP job traces at 70\% to 100\% load. The top row shows the total number of fragmented DP replica groups; the bottom row shows the maximum number of DP flows (rings) per ToR, with values exceeding 3 indicating congestion.

At 70\% load, fragmentation remains near zero except for two brief bursts, with little difference among algorithms—the cluster scheduler alone suffices. At 80\% load, ECMP suffers high fragmentation due to an early burst of arrivals, while \name and Crux (overlapping in the Figure) escape by completing jobs faster.
At 90\% load, the arrival rate overwhelms the cluster scheduler and routing-only schemes. ECMP and Crux reach average fragmentation degrees of 86 and 77, respectively, after the first 50 hours; Crux is slightly better because its more even traffic spreading completes jobs faster, freeing slots for better placement decisions. \name reduces total fragmentation by $4.85\times$ over Crux, only once exceeding 40 fragmentation degrees, and maintains per-rack fragmentation at or below 3—within Perfect routing's path-isolation guarantee. This enables \name to finish $1.25\times$ faster than Crux and $1.7\times$ faster than ECMP.
Finally, at 100\% load, ECMP and Crux become nearly indistinguishable, differing by an average of just 3 fragmentation degrees. \name oscillates between 3 and 4 per-rack fragmentation as it continuously defragments to keep pace with demand, finishing $1.33\times$ faster than Crux.
\enlargethispage{0.3\baselineskip}
\vspace*{-1mm}
\subsection{Performance across Oversub. Ratios}
\label{sec:eval_slo}

By keeping fragmentation in check, \name requires only a minimum number of uplinks, twice the number of DP flows per job, to guarantee path isolation. This implies \name's performance is solely determined by the number of uplinks per ToR, thereby enabling higher oversubscription than conventional routing-based approaches.

Figure~\ref{fig:slo} presents heatmaps of average and p99 slowdown for \name, Crux, ECMP, Perfect routing, and SGLB across a 2,048-GPU cluster. The $y$-axis varies the number of GPUs per rack from $16$ to $256$, and the $x$-axis varies the number of uplinks per rack by scaling the spine count until reaching full bisection bandwidth (entries along the diagonal). Packet spraying with full bisection bandwidth (not plotted) achieves a slowdown of 1 along this diagonal and serves as the baseline.

\begin{figure*}[t]
   \centering
    \captionsetup[subfigure]{aboveskip=-1pt,belowskip=-1pt}
    \begin{minipage}{0.3\textwidth}
    \centering
    \includegraphics[width=0.9\linewidth]{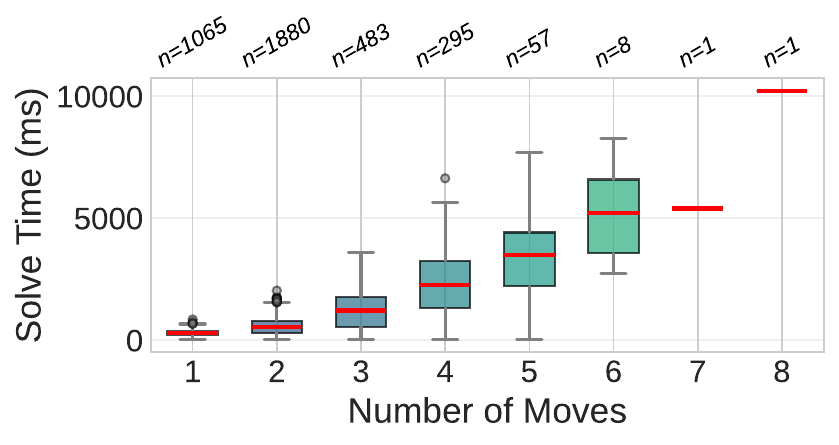}
    \caption{ILP solve time vs. number of migrations on a 1,024-GPU cluster. \normalfont Runtime scales  with move count, averaging 540~ms in the common two-move case.}
    \label{fig:ilp_runtime}
  \end{minipage}
  \hfill
  \begin{minipage}{0.32\textwidth}
    \centering
    \includegraphics[width=0.9\linewidth]{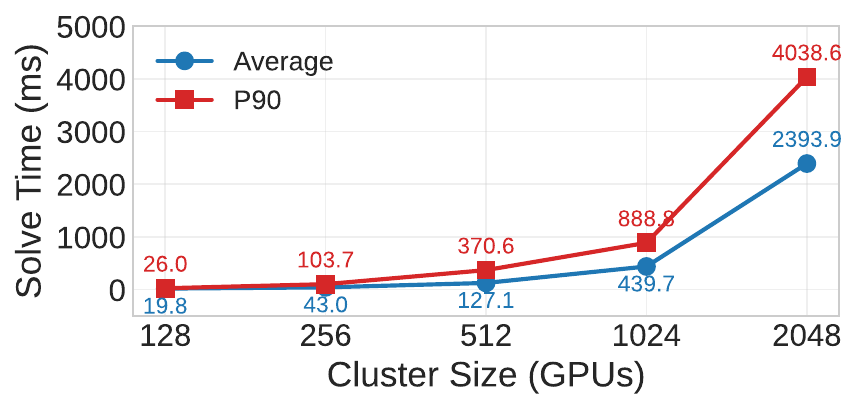}
    \caption{ILP solve time across cluster sizes from 128 to 2,048 GPUs. \normalfont p90 solve time stays under 5 s at 2,048 GPUs, well within practical planning budgets.}
    \label{fig:ilp_cluster_runtime}
  \end{minipage}
  \hfill
  \begin{minipage}{0.33\textwidth}
        \centering
    \includegraphics[width=0.9\linewidth]{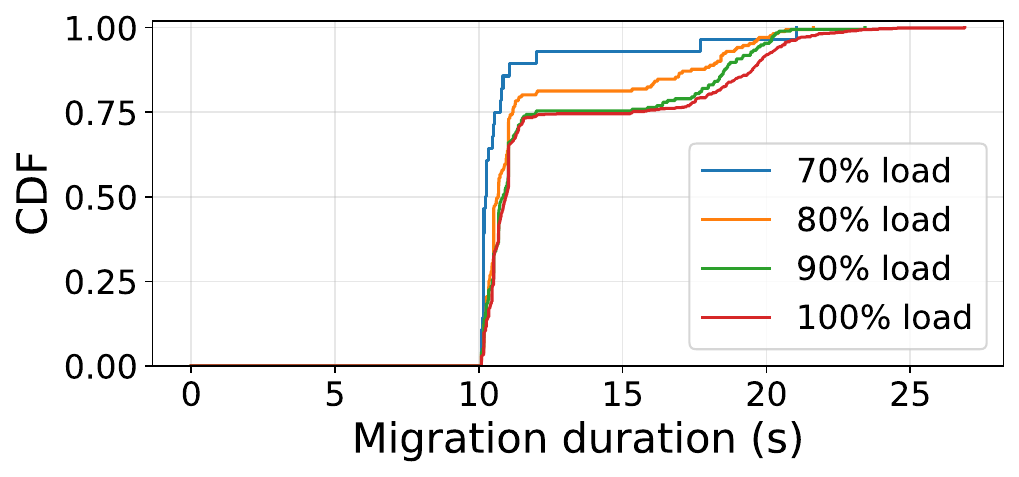}
    \caption{CDF of migration duration at different loads. \normalfont We assume 10~seconds of system overhead. Remaining time is due to network transfers and iteration synchronization.}
    \label{fig:migration_duration}
  \end{minipage}
   \vspace*{-3mm}
\end{figure*}

\name attains p99 performance within 5\% of packet spraying at an oversubscription ratio of 16:1, with average slowdown below $0.3$\% regardless of how many additional links are provisioned. Its only requirement is that the minimum number of uplinks is present; beyond that, it is nearly agnostic to the oversubscription ratio. Other systems require substantially more bandwidth. At $256$ GPUs per rack, SGLB, the best competing scheme, needs $4\times$ the uplinks to match \name's p99 performance. ECMP performs worst, with an average slowdown of up to $50$\% even at full bisection.

Notably, competing schemes do not maintain consistent performance even at a fixed oversubscription ratio. Crux has a $10$\% p99 slowdown at 4:1 oversubscription with $256$ GPUs per rack but $27$\% at the same ratio with $128$ GPUs per rack. The oversubscription ratio alone does not determine performance: at a fixed cluster size, fewer GPUs per rack means more racks, which causes more jobs to span rack boundaries and increases fragmentation. At the same time, each rack has fewer uplinks to absorb this fragmentation. Hence, what determines routing performance is the combination of fragmentation and oversubscription ratio. By controlling fragmentation, \name is the only system whose performance depends solely on the number of uplinks.

\subsection{Migration Solver Scalability}
\label{sec:ilp_results}
\label{sec:eval_migration_moves}
We evaluate the efficiency of \name's ILP solver, examining the number of migrations required to defragment the cluster and the solver's runtime characteristics. These results are critical for establishing that \name's migration planning is lightweight enough to be invoked frequently in practice.

\begin{figure}[t]
    \centering
    \includegraphics[width=.95\linewidth]{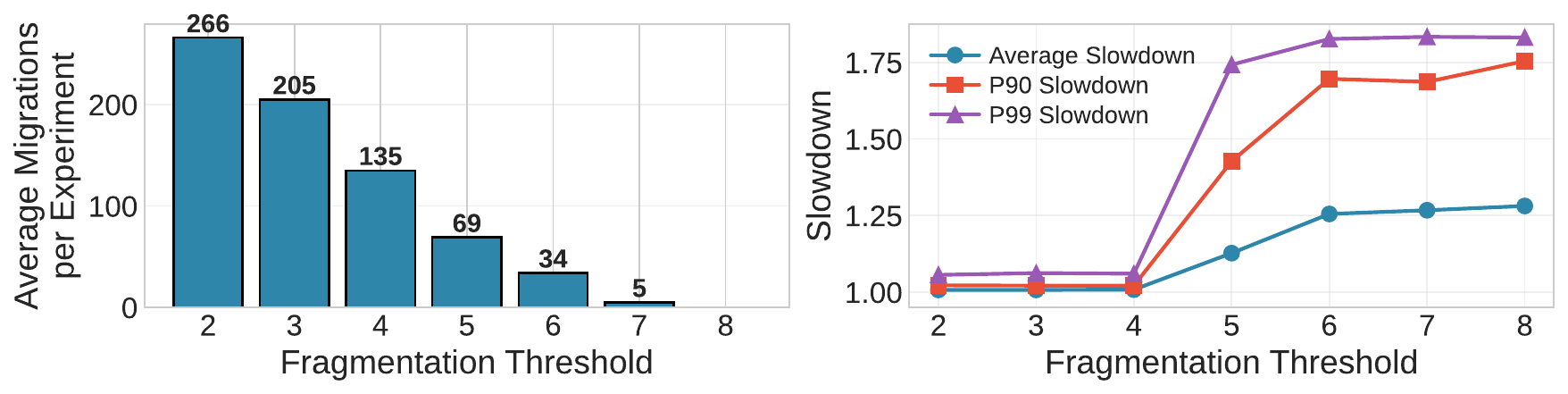}
    \caption{Number of migrations and job slowdown vs. fragmentation threshold on a 1,024-GPU cluster. \normalfont A threshold matching the number of uplinks achieves near-optimal performance with diminishing returns from further reduction.}
    \label{fig:lambda}
    \vspace*{-2mm}
\end{figure}

Figure~\ref{fig:ilp_moves} presents the distribution of required job migration (number of moves) to lower the fragmentation below two on clusters of varying sizes across five traces each. The distribution shifts as cluster size increases, but the average remains consistently between $1.7$ and $1.8$. Approximately $80\%$ of the weight is concentrated on one or two moves across all scenarios, and move counts greater than $5$ are rare, occurring in fewer than $1\%$ of the migrations for each cluster.

Figure~\ref{fig:ilp_runtime} shows the distribution of the ILP's solve time versus the number of moves across 5 traces on a 1,024-GPU cluster. The most common case of two moves averages 540~ms and never exceeds 4~s. Solve time scales roughly linearly with the number of moves, reaching 6.7~s at 6 moves. Cases requiring 7 or 8 moves each occurred only once and are thus not statistically meaningful. Overall, these results confirm that the ILP's runtime is governed primarily by the number of moves (\S\ref{sec:ring}).

Finally, Figure~\ref{fig:ilp_cluster_runtime} plots the ILP's solve time across cluster sizes ranging from $128$ to $2,048$ GPUs. For clusters of $1,024$ GPUs or fewer, the average solve time remains under $0.5$~s with a p90 below $1$~s. At $2,048$ GPUs, solve time increases by over $5\times$, with an average of $2.39$~s and a p90 of $4.04$~s. Even so, this remains well within practical limits: migrations are triggered infrequently and can tolerate seconds of planning latency, making the ILP viable even at the large cluster sizes.

\subsection{Migration Overhead}
\label{sec:migration_overhead}
This section analyzes the migration overhead required to keep fragmentation below the threshold.
Figure~\ref{fig:migration_duration} presents CDFs of migration duration across traces on a $1,024$-GPU cluster at $70$\% to $100$\% load. All distributions start at $10$ seconds, reflecting the baseline latency for workers to resume after migration completes.

Across all loads, the median migration duration remains below $11$ seconds, indicating that most migrations add little overhead beyond the baseline latency. The longest observed migration takes $26.89$ seconds; these tail cases arise from dependency chains in which a job must wait for downstream migrations to complete before it occupies the target slot. Even so, within $30$ seconds of initiating migrations, the cluster is fully restored to a congestion-free state.

The per-job overhead is also modest. No job in our experiments is migrated more than 6 times over its lifetime, resulting in a worst-case downtime of 115 seconds. Compared to that job's ideal runtime of 5.34 hours, migration constitutes less than 0.6\% overhead.

 \begin{figure}[t]
    \centering
    \begin{subfigure}[b]{0.48\columnwidth}
        \centering
        \includegraphics[width=\linewidth]{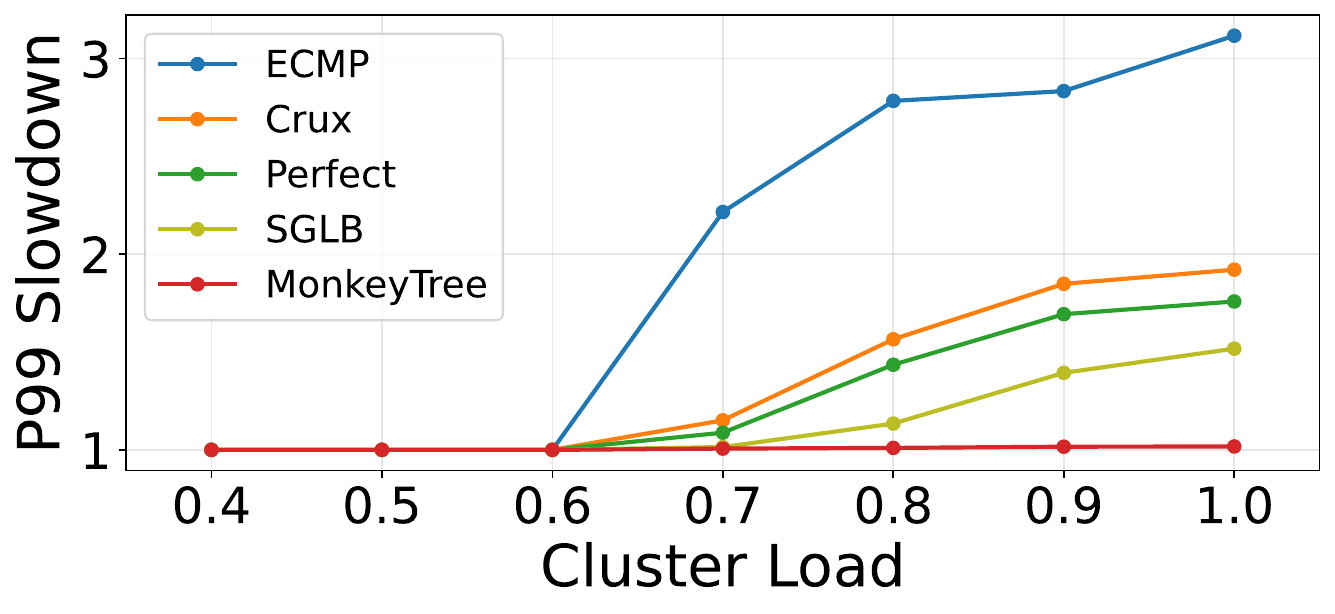}
        \vspace*{-5mm}
        \caption{Rail-optimized}
        \label{fig:rail}
    \end{subfigure}
    \hfill
    \begin{subfigure}[b]{0.48\columnwidth}     
        \centering
        \includegraphics[width=\linewidth]{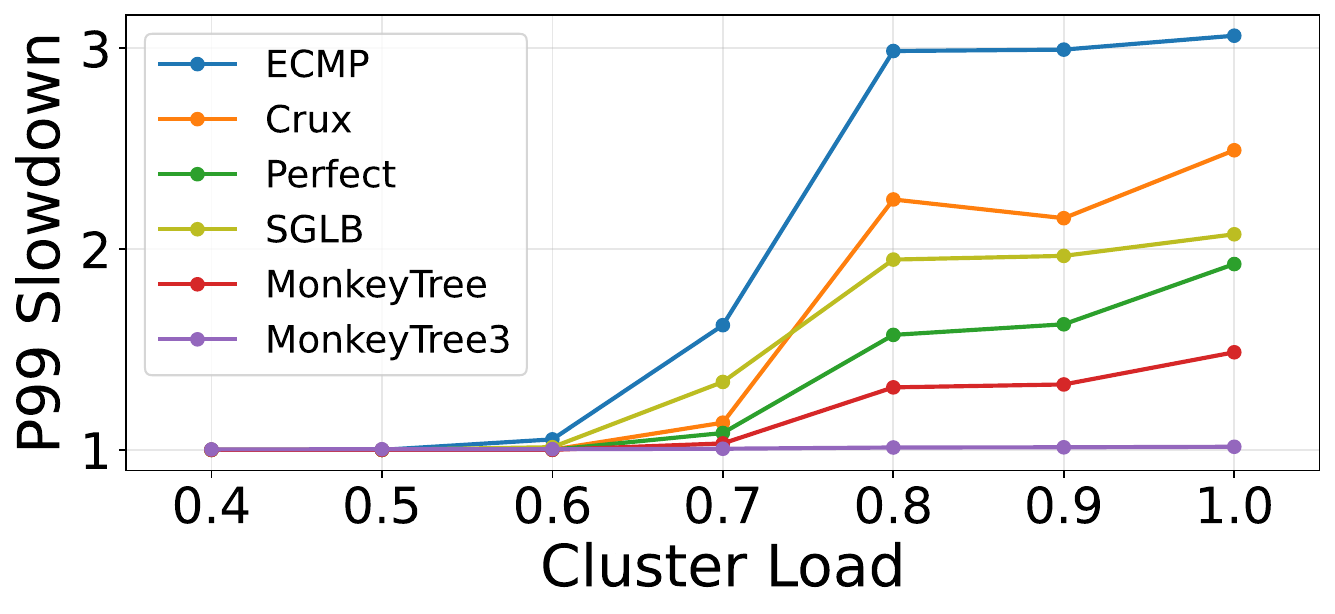}
        \vspace*{-5mm}
        \caption{3-tier Clos}
        \label{fig:fat}
    \end{subfigure}
    \vspace*{-1mm}
    \caption{Performance across rail-optimized and 3-tier Clos topologies. \normalfont \sys keeps p99 slowdown below 2\% on the rail-optimized topology, while MonkeyTree3 achieves the same on the 3-tier Clos by controlling fragmentation across both network levels.}
\end{figure}

\vspace*{-0.5mm}
\subsection{Threshold Sensitivity}
\label{sec:eval_lambda}

Section~\ref{sec:fragmentation} established that a per-rack fragmentation degree at or below the number of uplinks suffices for Perfect routing to guarantee path isolation, and that a fragmentation degree of two is always feasible. The fragmentation threshold controls when \name triggers migrations: \name migrates jobs only when a rack's fragmentation exceeds this threshold. Setting it above the number of uplinks reduces migrations at the cost of additional congestion. This section evaluates the tradeoff between this threshold and the resulting migration count and slowdown, using a trace of $1.2$k pure DP and FSDP jobs on a $1,024$-GPU cluster with four spine switches.

Figure~\ref{fig:lambda} plots the number of migrations alongside the average, p90, and p99 slowdown. At a threshold of $4$, matching the number of uplinks, p99 slowdown is within $6$\% of ideal, and the average is within $1$\%. This is slightly worse than the results in \S\ref{sec:eval_job_isolation}, as a workload composed entirely of FSDP and DP jobs with few uplinks presents a more constrained setting than the full workload mix. Decreasing the threshold below four yields less than $0.5$\% additional improvement: \textit{Perfect routing} already guarantees path isolation at a fragmentation threshold of four, so further migrations cannot eliminate elephant-flow collisions. A lower threshold only reduces collisions between elephant flows and in-flight migration traffic, yielding a marginal benefit offset at the cost of up to $2\times$ more migrations.

\subsection{Rail-optimized and 3-tier Clos}
In this section, we evaluate \name's performance across different types of topologies.
Figure~\ref{fig:rail} shows the p99 job completion slowdown of running traces of varying load across a 1,024 GPU cluster organized in a rail-optimized topology. Each rail is equipped with two uplinks to the spine layer. The results are similar to those observed in the 2-tier Clos topology (Figure~\ref{fig:slowdown}): \sys maintains a p99 slowdown of less than 2\% even under the highest load. The next-best system, SGLB, has a 51\% p99 slowdown at the highest load.

Figure~\ref{fig:fat} plots the p99 job completion slowdown of traces across a 1,024 GPU 3-tier Clos topology. Each ToR has 16 uplinks to the aggregation layer, and each aggregation switch has 16 uplinks to the core switch layer. MonkeyTree3, \sys's variant that addresses fragmentation on both levels, achieves the best performance, keeping the p99 slowdown within 2\% of the ideal. In contrast, the red \name line, which runs \sys only at the pod level, has a 48\% p99 slowdown due to the lack of cross-pod contention control. Unlike prior experiments, the next best system is \textit{Perfect-routing}, which reduces the p99 slowdown over SGLB by nearly 8\% at the highest load. This is because in SGLB, each switch makes independent decisions, which becomes more difficult as the number of switch layers and the number of switches per flow increase.



\begin{figure}[t]
    \centering
    \includegraphics[width=0.85\linewidth]{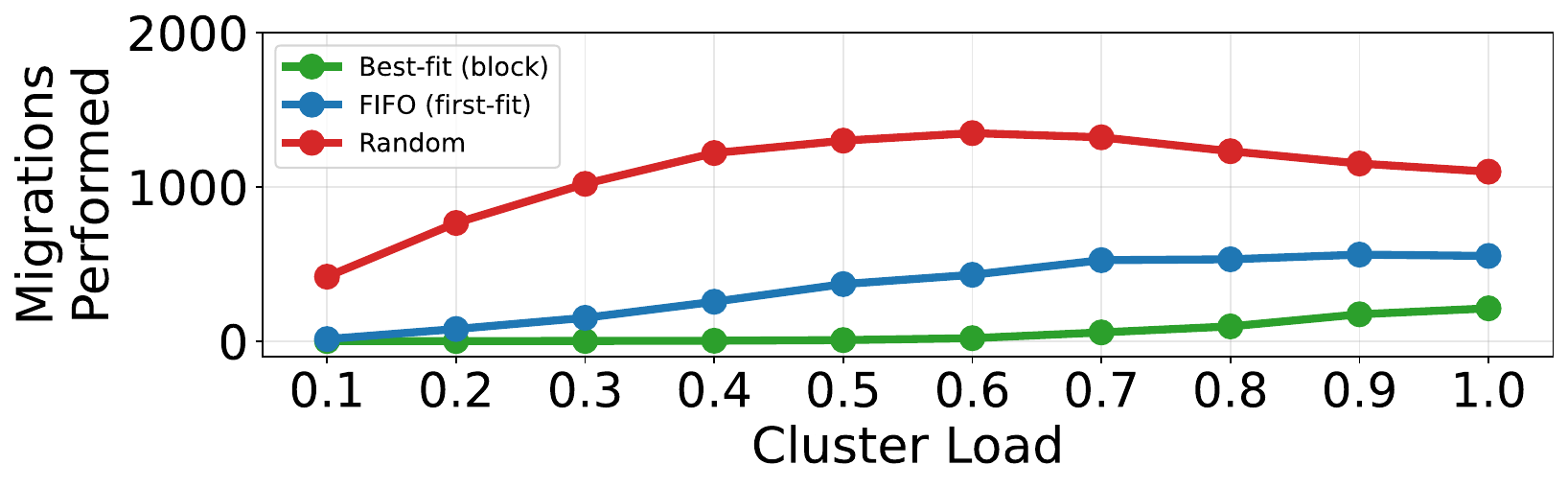}
    \vspace*{-2mm}
    \caption{Impact of front-end schedulers. \normalfont Best-fit reduces migrations compared to FIFO and Random.}
    \label{fig:scheduler_ablation}
\end{figure}

\begin{figure}[t]
    \centering
    \includegraphics[width=0.99\linewidth]{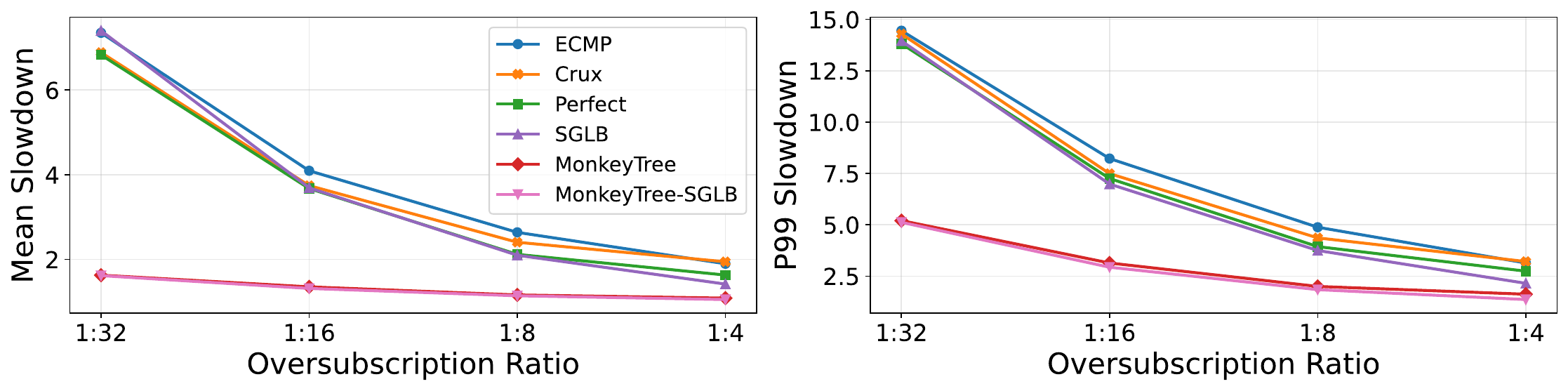}
    \caption{EP-dominated workloads. \normalfont \sys-SGLB performs reduces p99 slowdown by up to 25\% over \sys with Perfect-routing, because dynamic load-aware routing is more effective for dense EP traffic.}
    \label{fig:moe_spine}
\end{figure}

\subsection{Front-end Scheduler Ablation Study}
This section performs an ablation study of different front-end schedulers to assess their impact on \sys. We compare three schedulers on traces of 2.4k jobs with varying load: Best-fit, used in the rest of the evaluation; FIFO, which scans for free GPUs from left to right; and Random, which assigns jobs randomly among available slots.

Best-fit reduces the number of migrations by $2.6\times$ at the highest load compared to FIFO, the next-best scheduler. This is because Best-fit makes locally optimal placement decisions that keep jobs on as few racks as possible. Random induces the most migrations, peaking at $1.3$k migrations at 60\% load, since its placements often have poor locality. At higher loads, however, cluster churn increases, and exiting jobs are quickly replaced by new ones, leaving fewer machines available for Random to choose from. As a result, Random makes fewer poor placement decisions. Since migrations also require jobs to be present and fragmentation to remain below the thresholds, Random peaks at 60\% load rather than at either extreme.
Despite these differences, all three schedulers achieve similar performance, with p99 slowdown within 1\%. With 2.4k jobs and at most 1.3k migrations, roughly every other job is migrated once, and the migration cost remains small relative to total training time.

\subsection{EP Dominant Workloads}
\label{eval:ep}
This section evaluates \sys and other systems on EP-dominated workloads. We use a 1,024-GPU spine-tree topology and run traces with varying load, where every job’s network demand is dominated by expert parallelism, with at least $10\times$ more All-to-All traffic than ring-based traffic. Because simulating All-to-All is computationally expensive, we use 100 jobs per trace and initialize the cluster from sampled placements of normal runs: \sys placements for \sys variants, and ECMP placements for all other systems.
Figure~\ref{fig:moe_spine} shows the results. \sys-SGLB performs best overall, reducing p99 job completion slowdown by up to $19\%$ over \sys with \textit{Perfect-routing}. For All-to-All traffic, \sys cannot provide the isolation guarantee used by \textit{Perfect-routing}; as a result, SGLB's dynamic, load-aware rerouting outperforms static routing. SGLB is also the best non-\sys system, reducing p99 slowdown by $28\%$ over \textit{Perfect-routing} at the highest load. \textit{Perfect-routing} and Crux provide smaller gains over ECMP in this setting, reducing p99 slowdown by at most $1.24\times$. This is because All-to-All traffic creates many flows, making randomized ECMP more effective than it is for sparse ring-based collectives.

\section{\name Prototype} \label{sec:prototype}
We prototype \sys on a testbed of five servers, each with one Nvidia A100 GPU~\cite{a100} and 100~Gbps ConnectX-5 NIC~\cite{roce}, to demonstrate practicality and quantify migration overhead.
\para{RDMA-Based Checkpointing.} We extend PyTorch's Distributed Checkpoint (DCP)~\cite{pytorch_dcp} with custom \texttt{StorageReader} and \texttt{StorageWriter} interfaces that perform RDMA transfers instead of filesystem I/O for migration. During migration, the source worker writes GPU tensors to the CPU DRAM of a remote staging server using GPUDirect RDMA. Targeting CPU DRAM avoids contention for GPU memory at the destination and enables pipelined loading, allowing tensor transfers to overlap with model initialization. The replacement worker constructs ShardedTensor templates matching the original layout and accepts RDMA writes to retrieve the checkpoint. A lightweight protocol coordinates shard availability without blocking the critical path. \looseness=-1

\begin{figure}[t]
    \centering
    \includegraphics[width=0.85\linewidth]{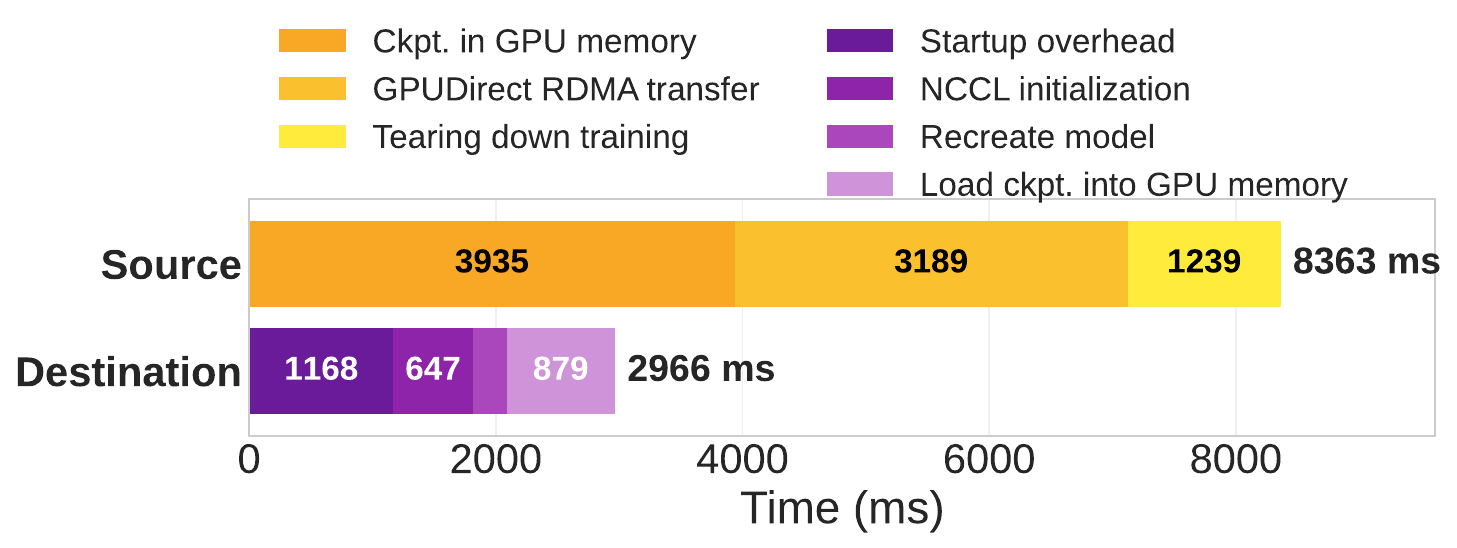}
    \vspace*{-1mm}
    \caption{Timing breakdown of migrating a single Llama3-8B worker on the testbed. \normalfont Checkpoint composition dominates at $3.93$~s, and system overhead besides network transfer is $9.02$~s, validating the $10$ s migration budget assumed in simulation.}
    \label{fig:migration_overhead}
\end{figure}

Each \sys daemon creates a migration listener at cluster start with a pre-allocated buffer to hold the checkpoint. When the controller triggers a migration, it sends the migrating job a signal with the worker that is migrating and its destination. The daemons on the migrating jobs wait for the current training iteration to complete before initiating the RDMA checkpoint. Once the transfer finishes, the job synchronizes, loads from the checkpoint, and resumes training.

\para{Prototype microbenchmark.} Figure~\ref{fig:migration_overhead} presents the timing breakdown of migrating a single worker in a $4$-node Llama3-8B training instance using our prototype. The top column plots overhead on the sender side (source of migration), while the bottom one plots the receiver side (destination of migration). The dominant cost is checkpointing at $3.93$~s. \sys-daemon transfers the $12$~GB checkpoint shard (model weights and optimizer state) over RDMA in $3.19$~s end-to-end. The transfer itself accounts for only $1.07$~s at $90$~Gbps; the remainder is serialization and memory management overhead. Source server cleanup adds $1.24$~s, though this does not block the migrated model from resuming. On the destination, NCCL initialization, model loading, and other startup costs total $2.97$~s. Excluding the RDMA transfer's network time, the total overhead is $9.02$~s, validating the $10$~s migration budget assumed in simulations. \looseness=-1

\section{Discussion}\label{sec:discussion}

\para{Scaling \sys to larger clusters.}
\sys's ILP exhibits p90 solve times below five seconds for clusters of up to 2048 GPUs.
We divide larger clusters into independent partitions, each managed by a separate \name~controller, thereby keeping ILP solve times tractable.
A second-tier \name~controller then treats each sub-cluster as a virtual rack and manages migrations across them.



\para{Limitations of \sys on EP and CP.}
Our formulation assumes collective traffic is dominated by sparse data-parallel rings, which does not hold for all parallelism strategies.

First, EP without expert replication introduces dense All-to-All traffic that can be on par with DP traffic, violating this assumption.
As discussed earlier, \sys still helps by reducing fragmentation, but it cannot guarantee congestion-free execution for EP-dominated workloads.
Our evaluation in \S\ref{eval:ep} confirms this behavior: \sys is most effective when paired with dynamic load-aware routing, reducing p99 slowdown by up to $19\%$ over static routing.
Models with modest top-$k$ gates and expert counts remain more tractable, and expert replication further confines the All-to-All to a many-to-many pattern.\looseness=-1

Second, CP partitions long sequences across workers and requires both data-parallel and context-parallel rings.
This makes placement an NP-hard min-cut variant.
While our ILP can be extended to handle EP and CP, it does not provide congestion-free guarantees.
We leave EP-aware topology optimization under varying replication strategies and CP-aware placement to future work.

\para{Other types of collectives.} Ring-based collectives remain widely used, but other algorithms, such as tree-based collectives, are increasingly being adopted. The core principle behind \sys is not specific to rings: bounding placement fragmentation improves locality and limits network contention across collective algorithms. While different collectives may require different placement rules to obtain tight bounds, the objective remains similar: because ML training jobs often have highly symmetric communication patterns, low-fragmentation placements remain effective for many non-ring collectives.

\para{\sys under lower oversubscription.} \sys supports highly oversubscribed topologies through guarantees independent of the oversubscription ratio. Even in less oversubscribed deployments, \sys reduces traffic to upper network layers, providing benefits beyond isolation for sparse ring flows. For denser patterns such as EP, the additional links help balance load; for sparse traffic, they improve fault tolerance under link failures.

\section{Related Work}\label{sec:related}

\para{Mitigating congestion in multi-tenant training clusters.}
Prior work on scheduling multi-tenant training jobs follows two directions.
The first temporally shares network links by interleaving communication from competing jobs into idle periods during computation~\cite{rajasekaran_hotnets, mltcp, cassini}.
This approach relies on favorable computation-communication overlap patterns, which are not always present.
The second spatially distributes traffic across multiple paths via routing or adaptive load balancing~\cite{crux, sglb}.
\name~employs similar load-balancing mechanisms but also controls fragmentation to bound aggregate network load.
More recent systems combine temporal and spatial techniques~\cite{harmonics, foresight}; however, none provide explicit guarantees on maximum network load.

\para{Congestion control in datacenter networks.}
A vast literature addresses datacenter congestion through congestion control protocols~\cite{congestion_control, cubic, reno, hpcc, dctcp}, flow scheduling~\cite{pfabric, hedera, PDQ, UPS, homa}, and coflow abstractions~\cite{coflows, coflow1, sincronia, varys, coda}.
\name~takes a complementary approach, managing congestion through placement and migration rather than network-layer mechanisms.

\para{CPU migration.} CPU migration is well studied for load balancing in shared clusters~\cite{nsdi_live_migration_core_1, migration_core_2, sandpiper}, though later work adopted more conservative policies as datacenter-scale deployments exposed migration overheads~\cite{vm_risky_1, vm_risky_2}.
The most closely related line of work jointly optimizes placement and routing, balancing migration cost and network congestion~\cite{joint_migration_routing, joint_migration_routing_2, joint_placement_routing_3}.
\name~decouples this tradeoff: operators specify a target network load, and \name~computes the minimum number of migrations required to meet it.

\para{GPU migration.} Prior work on GPU migration falls into two categories.
The first focuses on efficient migration mechanisms: systems such as TrainMover enable live migration with minimal training disruption~\cite{trainmover, dcuda}.
\name~instead adopts a lightweight checkpoint-and-restore approach that buffers state in CPU memory rather than on disk; these techniques are orthogonal and would reduce migration overhead in \name.
The second targets GPU resource defragmentation and fine-grained allocation at the granularity of individual GPUs or MIG partitions~\cite{gandiva, defrag_gpu_scheduler}.
\name~targets large-scale training jobs that occupy entire servers, where contention manifests in the network rather than at the device level.

\para{Elastic training.}
Elastic training allows workers to dynamically join and leave a training job to improve fault tolerance and utilization~\cite{pipetransformer, parcae, elaswave}.
Such mechanisms complement \name's migrations by reducing the associated disruption. 
\section{Conclusion}
\label{sec:conclusion}
We presented \sys, a migration-based defragmentation system that mitigates network congestion in multi-tenant GPU clusters. Our evaluation shows that \sys improves average job completion time by $14$\% over the next-best baseline and keeps p99 slowdown to $5$\% at 16:1 oversubscription on $2,048$ GPUs, with a $9.02$~s migration overhead per worker via in-memory RDMA checkpoint-and-restore.

\section*{Acknowledgments}
We thank the anonymous shepherd for their insightful comments. Thank you to Hari Balakrishnan, Jeremy Carin, Om Chabra, Snehadeep Gayen, Sanjoli Narang, and Benny Rubin for their comments on earlier versions of this paper. This research was supported by NSF FMitF-2421734, NSF CAREER-2144766, NSF PPoSS-2217099, NSF CNS-2211382, Jane Street Group LLC, and Sloan fellowship FG-2022-18504.

\newpage

\label{bodypage}
\bibliographystyle{ACM-Reference-Format}
\bibliography{reference}

\appendix
\clearpage
\section*{Appendix}

\section{Casting \name's Optimization as an ILP}
\label{appendix:full_ilp}
In this section, we detail how to cast the core algorithmic problem of \name as an ILP.

The problem is about changing edge weights in a fully connected bipartite edge graph. $S$ is the set of left nodes, $T$ is the set of right nodes, and $E$ is the set of edges. For each fragmented job in the cluster, we add a vertex $s$ for it to $S$. For each ToR switch in the cluster, we add a vertex $t$ to $T$. We add an edge between each job-ToR pair and set its weight to the number of jobs that job $s$ has on ToR $t$.

The optimization objective is to shift these edge weights around minimally such that all constraints are met. Table~\ref{table:full_ilp} defines all decision variables and constants. Below we specify the full ILP.

\begin{table}[t]
    \centering
    \scriptsize
    \setlength{\tabcolsep}{8pt}
    \begin{tabular}{@{}c|ccc@{}}
    \toprule
    \textbf{Variable}       & \textbf{Type} & \textbf{Value Range} & \textbf{Explanation} \\ \midrule
    \textbf{$w(s, t)$} & Decision Variable & $\geq 0$ & Assignment of workers to ToRs.\\
    \textbf{$x(s, t)$} & Decision Variable & Binary & Job $s$ is present on ToR $t$. \\
    $d(s, t)$ & Decision Variable & $\geq 0$ & Deviation between $w(s, t)$ and $w_0(s, t)$.\\
    $y(s, t)$ & Decision Variable & Binary & Job $s$ is \textit{only} present on ToR $t$. \\
    $w_0(s, t)$ & Input & $\mathbb{Z}_{\geq 0}$ & Initial assignment of workers to ToRs. \\
    $R_s$ & Input & $\mathbb{Z}_{\geq 1}$ & Number of workers for job $s$. \\
    $\rho_s$ & Input & $\mathbb{Z}_{\geq 1}$ & Number of independent rings for job $s$ (fragmentation weight). \\
    $\beta_t$ & Input & $\mathbb{Z}_{\geq 1}$ & Capacity ToR $t$ has. Excludes non-fragmented jobs.\\
    $\gamma_t$ & Input & $\mathbb{Z}_{\geq 0}$ & Fragmentation threshold. \\
    $U_{s, t}$ & Input & $\mathbb{Z}_{\geq 1}$ & Big-$M$ upper bound on $w(s, t)$, $\min(R_s, \beta_t)$. \\
    \bottomrule
    \end{tabular}
    \caption{Decision variables \& inputs for the ILP.}
    \label{table:full_ilp}
\end{table}


\setcounter{equation}{0}
\begin{align}
\min_{w,\,x,\,d,\,y} \quad & \sum_{(s, t) \in S \times T} d(s, t) \label{eq:obj} \\
\text{s.t.} \quad & \sum_{t} w(s, t) = R_s & \forall s \in S \label{eq:rowsum} \\
& \sum_{s} w(s, t) \leq \beta_t & \forall t \in T \label{eq:capacity} \\
& w(s, t) \leq U_{s, t} \, x(s, t) & \forall (s, t) \in S \times T \label{eq:activation} \\
& \sum_{s} \rho_s \, x(s, t) - \sum_{s} \rho_s \, y(s, t) \leq \gamma_t & \forall t \in T \label{eq:frag} \\
& d(s, t) \geq w(s, t) - w_0(s, t) & \forall (s, t) \in S \times T \label{eq:dev1} \\
& d(s, t) \geq w_0(s, t) - w(s, t) & \forall (s, t) \in S \times T \label{eq:dev2} \\
& y(s, t) \leq x(s, t) & \forall (s, t) \in S \times T \label{eq:ylink} \\
& \sum_{t' \neq t} x(s, t') \leq (|T| - 1)\,(1 - y(s, t)) & \forall s \in S, \; \forall t \in T \label{eq:ysingle}
\end{align}

Equation 1 is the objective and represents minimizing the distance between the final and initial state. Equation 2 means that the total number of workers to a job must add up to that job's demand. Equation 3 enforces that the amount of workers assigned to each ToR does not exceed its capacity. Equation 4 sets $x(s, t)$. It must be set to $1$ for job $s$ to have any workers on ToR $t$. Equation 5 checks the fragmentation constraint. The first summation is the total amount of fragmentation from jobs present on the ToR, and the second summation subtracts any jobs that were defragmented during this process. Equations 6 and 7 ensure the $d(s, t)$ variable used in the objective measures the final and initial state difference. Equation 8 ensures that $y(s, t)$ can only be $1$ if $x(s, t)$ is. Equation 9 lets $y(s, t)$ be $1$ if the job $s$ is only located on ToR $t$, excluding all other ToRs $t'$. 

\end{document}